\documentclass[11pt]{article}

\usepackage[letterpaper,margin=1.00in]{geometry}
\usepackage{amsmath, amssymb, amsthm, amsfonts}
\usepackage{bbm}

\usepackage{ifthen}
\usepackage{tikz}
\usetikzlibrary{positioning,decorations.pathreplacing}
\usepackage{ bbold }
\usepackage{cite}
\usepackage{appendix}
\usepackage{graphicx}
\usepackage{color}
\usepackage{color}
\usepackage{algorithm}
\usepackage{algorithm}
\usepackage[noend]{algpseudocode}
\usepackage{epstopdf}
\usepackage{wrapfig}
\usepackage{paralist}
\usepackage{wasysym}
\usepackage[textsize=tiny]{todonotes}

\usepackage{framed}
\usepackage[framemethod=tikz]{mdframed}
\usepackage[bottom]{footmisc}
\usepackage{enumitem}
\setitemize{noitemsep,topsep=3pt,parsep=3pt,partopsep=3pt}
\usepackage[font=small]{caption}
\usepackage{xspace}
\usepackage{thmtools}
\usepackage{thm-restate}

\newtheorem{theorem}{Theorem}[section]
\newtheorem{lemma}[theorem]{Lemma}
\newtheorem{meta-theorem}[theorem]{Meta-Theorem}

\newtheorem{remark}[theorem]{Remark}
\newtheorem{corollary}[theorem]{Corollary}
\newtheorem{proposition}[theorem]{Proposition}
\newtheorem{observation}[theorem]{Observation}

\definecolor{darkgreen}{rgb}{0,0.5,0}
\usepackage{hyperref}
\hypersetup{
    unicode=false,          % non-Latin characters in Acrobat’s bookmarks
    colorlinks=true,        % false: boxed links; true: colored links
    linkcolor=red,          % color of internal links (change box color with linkbordercolor)
    citecolor=darkgreen,        % color of links to bibliography
    filecolor=magenta,      % color of file links
    urlcolor=cyan           % color of external links
}
\usepackage[capitalize, nameinlink]{cleveref}
\crefname{theorem}{Theorem}{Theorems}
\crefname{proposition}{Proposition}{Propositions}
\crefname{observation}{Observation}{Observations}
\Crefname{lemma}{Lemma}{Lemmas}

\algnewcommand\algorithmicswitch{\textbf{switch}}
\algnewcommand\algorithmiccase{\textbf{case}}

\algdef{SE}[SWITCH]{Switch}{EndSwitch}[1]{\algorithmicswitch\ #1\ \algorithmicdo}{\algorithmicend\ \algorithmicswitch}%
\algdef{SE}[CASE]{Case}{EndCase}[1]{\algorithmiccase\ #1}{\algorithmicend\ \algorithmiccase}%
\algtext*{EndSwitch}%
\algtext*{EndCase}%
\newcommand{\congest}{$\mathsf{CONGEST}$\xspace}
\newcommand{\local}{$\mathsf{LOCAL}$\xspace}

\newcommand{\poly}{\operatorname{\text{{\rm poly}}}}

\newcommand{\lev}{\textrm{lev}}
\newcommand{\fC}{{\mathcal{C}}}
\newcommand{\id}{{\textrm{id}}}

\renewcommand{\paragraph}[1]{\vspace{0.15cm}\noindent {\bf #1}:}

\newcommand{\FullOrShort}{full}

\ifthenelse{\equal{\FullOrShort}{full}}{
  
  \newcommand{\fullOnly}[1]{#1}
  \newcommand{\shortOnly}[1]{}

  }{

    \newcommand{\fullOnly}[1]{}
    \newcommand{\IncludePictures}[1]{}
   
  }

\begin{document}
\title{Improved Deterministic Network Decomposition}
\author{
Mohsen Ghaffari\\ \small{ETH Zurich} \\ \small{ghaffari@inf.ethz.ch}
\and Christoph Grunau \\ \small{ETH Zurich} \\ \small{cgrunau@student.ethz.ch}
\and V\'{a}clav Rozho\v{n} \\ \small{ETH Zurich} \\ \small{rozhonv@ethz.ch}
}

\date{}
\maketitle

\begin{abstract}
    Network decomposition is a central tool in distributed graph algorithms. We present two improvements on the state of the art for network decomposition, which thus lead to improvements in the (deterministic and randomized) complexity of several well-studied graph problems.

    \begin{itemize}
        \item[-] We provide a deterministic distributed network decomposition algorithm with $O(\log^5 n)$ round complexity, using $O(\log n)$-bit messages. This improves on the $O(\log^7 n)$-round algorithm of Rozho\v{n} and Ghaffari [STOC'20], which used large messages, and their $O(\log^8 n)$-round algorithm with $O(\log n)$-bit messages. This directly leads to similar improvements for a wide range of deterministic and randomized distributed algorithms, whose solution relies on network decomposition, including the general distributed derandomization of Ghaffari, Kuhn, and Harris [FOCS'18].
        \medskip
        
        %\item[-] As an example application of this improvement and using some other ideas, we get a deterministic algorithm for the maximal independent set problem with round complexity $O(\log^5 n)$, and using $O(\log n)$-bit messages. This improves on the $O(\log^8 n)$ algorithm of Rozho\v{n} and Ghaffari.
        
        %\medskip
        \item[-] One drawback of the algorithm of Rozho\v{n} and Ghaffari, in the \congest  model, was its dependence on the length of the identifiers. Because of this, for instance, the algorithm could not be used in the shattering framework in the \congest  model. Thus, the state of the art randomized complexity of several problems in this model remained with an additive $2^{O(\sqrt{\log\log n})}$ term, which was a clear leftover of the older network decomposition complexity [Panconesi and Srinivasan STOC'92]. We present a modified version that remedies this, constructing a decomposition whose quality does not depend on the identifiers,  and thus improves the randomized round complexity for various problems.
    \end{itemize}
\end{abstract}

{
    \setcounter{page}{0}
    \thispagestyle{empty}
}
\newpage
{   
    \hypersetup{linkcolor=blue}
    \tableofcontents
    \setcounter{page}{0}
    \thispagestyle{empty}
}

\newpage

\section{Introduction and Related Work}
\label{sec:intro}
Network decomposition is a central tool in distributed graph algorithms that was first introduced in the seminal work of Awerbuch, Goldberg, Luby, and Plotkin~\cite{awerbuch89}. Currently, the complexity of a wide range of deterministic and randomized distributed algorithms for various local graph problems rests on the complexity of network decomposition. In this work, we present (quantitative and qualitative) improvements on the state of the art network decomposition algorithm.

\subsection{Background}
\paragraph{Distributed Model} We work with the standard synchronous message passing modeling of distributed algorithms on networks. The network is abstracted as an $n$-node graph and there is one processor on each node of the graph. Per round, each processor/node can send one message to each neighbor. If the message size is unbounded, the this is known as the \local model~\cite{linial1987LOCAL}. If the message size is bounded, to some $B$ bits, this is known as the \congest model; the typical assumption then is that $B=\Theta(\log n)$. Initially, nodes 
do not know the topology of the network $G$, except for potentially some estimates on basic global parameters such as the number of nodes $n$ (which is tight up to a polynomial). 
When discussing deterministic algorithms, we assume that each node has a unique $b$-bit identifier, and again the most typical case is to assume $b=\Theta(\log n)$. At the end of the algorithm, each node should know its own part of the output, e.g., its own color when coloring the vertices. The main measure of interest is the round complexity of the algorithm, i.e., the number of rounds until all nodes have finished their computation. 

\paragraph{Network Decomposition} A $(C, D)$ network decomposition of a graph $G=(V, E)$ is a partition of the vertices into disjoint clusters such that each cluster has diameter at most $O(D)$ and where clusters are colored with $O(C)$ colors in a way that adjacent clusters have different colors. A small subtlety is in the definition of the term ``diameter", according to which we can categorize decompositions into two types: (A) in a \textit{strong-diameter decomposition}, any two vertices of a cluster have distance $O(D)$ in the subgraph induced by that cluster, (B) in a \textit{weak-diameter decomposition}, any two vertices of a cluster have distance $O(D)$ in the base graph $G$.

For any $n$-node graph, there is an $(\log n, \log n)$ network decomposition, and this can be computed sequentially via a simple ball carving algorithm~\cite{Awerbuch-Peleg1990, linial93}. Network decomposition is immediately useful for distributed algorithms. As a simple example, given a $(C, D)$ network decomposition (even with weak-diameter), we can compute a maximal independent set (MIS) of the graph in $O(CD)$ rounds in the \local model, by simulating the corresponding sequential greedy algorithm, as follows: We process the colors one by one. Per color, each cluster gathers the topology of the cluster and its immediate neighborhood to the center of the cluster, in $O(D)$ rounds, and decides which vertices of the cluster can be added to the MIS. See \cite{ghaffari2017complexity} for a more general explanation of how one can transform a certain class of sequential algorithms (formally, in the $\mathsf{SLOCAL}$-model) to distributed algorithms in the \local model, using network decomposition, with an $O(CD)$ overhead in locality. See also \Cref{subsec:OtherRelated} for other related work.
% \vspace{4cm}

\subsection{State of the Art} 
\paragraph{Deterministic Algorithms} Awerbuch et al.~\cite{awerbuch89} gave an algorithm that deterministically computes $(C, D)$ strong-diameter network decomposition in $T$ rounds (even in the \congest model), where $C=D=T=2^{O(\sqrt{\log n \log\log n})}$. Panconesi and Srinivasan~\cite{panconesi-srinivasan} provided a variant of this deterministic algorithm (for the \local model) that improved the bounds to $C=D=T=2^{O(\sqrt{\log n})}$.  However, this $2^{O(\sqrt{\log n})}$ bound remained the state of the art complexity for network decomposition for over 25 years. It also remained the state of the art deterministic complexity for a long list of other fundamental graph problems whose solutions deterministic relied on network decomposition, including maximal independent set, $\Delta+1$ coloring, Lov\'{a}sz Local Lemma, etc, and which were known to admit $\poly(\log n)$ round randomized algorithms. This significant gap between randomized and deterministic algorithms was a central open problem in distributed graph algorithms; see, e.g., the open problems chapter of the 2013 book by Barenboim and Elkin book~\cite{barenboimelkin_book}. Surprisingly, it was also (a provable) bottleneck in the complexity of many randomized algorithms, as shown by Chang, Kopelowitz, and Pettie~\cite{chang16}. See \Cref{subsec:OtherRelated} for some other related work.

Recently, Rozho\v{n} and Ghaffari~\cite{RozhonG19} presented the first deterministic decomposition algorithm with poly-logarithmic parameters and complexity. Concretely, they obtained a $(\log n, \log^3 n)$ weak-diameter decomposition in $O(\log^7 n)$ rounds of the \local model or $O(\log^8 n)$ rounds of the \congest model. 
They also explained how this leads to a $(\log n, \log n)$ strong-diameter decomposition in $O(\log^8 n)$ rounds of the \local model. These results led to the first $\poly(\log n)$ round deterministic distributed algorithms for a wide range of local graph problems, as well as significant improvements for many randomized algorithms (in the shattering framework,  see, e.g.,~\cite{barenboim_symmbreaking, ghaffari2016MIS, chang2018optimal, chang2019coloring}).

\subsection{Our Contributions}
\label{sec:our_contributions}
Our contributions provide improvements on the result of Rozho\v{n} and Ghaffari~\cite{RozhonG19}, in two essentially-orthogonal directions:

\paragraph{Direction 1 -- Faster Decomposition, and Applications}
Our first contribution is to present a faster algorithm that also computes a qualitatively better network decomposition:

\begin{theorem}[Informal Version of \cref{thm:log5}]
\label{thm:log5informal}
There is a deterministic distributed algorithm, in the \congest model, that computes a $(\log n, \log^2 n)$ network decomposition in $O(\log^5 n)$ rounds.
\end{theorem}

This should be contrasted with the $(\log n, \log^3 n)$ network decomposition of Rozho\v{n} and Ghaffari~\cite{RozhonG19} that had a $O(\log^7 n)$ round complexity in the \local model and $O(\log^8 n)$ round complexity in the \congest model. As in their work, in the \local model, one can turn this into a strong-diameter $(\log n, \log n)$ network decomposition, in $O(\log^6 n)$ rounds.

Our faster $O(\log^5 n)$-round algorithm immediately leads to a similar round complexity improvement for all the applications of deterministic network decomposition. As concrete examples, we show how we can deterministically solve maximal independent set and $\Delta+1$ coloring problems in $O(\log^5 n)$ and $O(\log^6 n)$ rounds of the \congest model, respectively. These algorithms improve on the $O(\log^7 n)$-round \local model algorithms of Rozho\v{n} and Ghaffari~\cite{RozhonG19} for the \local model, as well as the $O(\log^8 n)$-round \congest-model algorithms of Censor-Hillel et al.~\cite{censor2017derandomizing, RozhonG19} for MIS and of Bamberger et al.~\cite{bamberger2020efficient} for coloring. We comment that these improvements, besides the new network decomposition, also use some other ideas for pipelining information in the \congest model to save an additional factor of $\log n$. 
\begin{restatable}{corollary}{deterministicmis}
There is a deterministic distributed algorithm, in the \congest model, that computes a maximal independent set in $O(\log^5 n)$ rounds.
\end{restatable}

\begin{restatable}{corollary}{deterministiccoloring}
There is a deterministic distributed algorithm, in the \congest model, that computes a $\Delta+1$ coloring, where $\Delta$ is an upper bound on the maximum degree, in $O(\log^6 n)$ rounds.
\end{restatable}

\paragraph{Direction 2 -- Identifier-Independent Decomposition, with Application}
One drawback of the construction of Rozho\v{n} and Ghaffari~\cite{RozhonG19} in the \congest model was that the quality of the obtained network decomposition depends on the length of the identifiers provided. For instance, in a network with $b$-bit identifiers --- and where thus $b$-bit messages are permitted---their algorithm computes a $(\log n, b^2 \log n)$ decomposition, in $O(b^4\log^3 n)$ rounds. This bad dependency on the length of the identifiers becomes a bottleneck in some applications: in particular, it was not possible to use their algorithm in the shattering framework for randomized algorithms with small messages, and the best known algorithm in the \congest-model remained with an $2^{O(\sqrt{\log \log n})}$ term in the round complexity, which was a clear remnant of the old $2^{O(\sqrt{\log n})}$ round complexity of deterministic network  decomposition~\cite{panconesi-srinivasan, ghaffari2019MIS}. We present a variant of their algorithm that computes a $(\log n, \log^3 n)$ decomposition in $O(\log^8 n + \log^7 n \log^* b)$ rounds  in the setting with $b$-bit identifiers and using $b$-bit messages. This is achieved by replacing the reliance of the construction's invariant on the bits of the identifiers by some semi-balanced $2$-coloring of the clusters, which is computed in the course of the construction.

Furthermore, we show that this second improvement is compatible with the first, in the sense that we can put the two ideas together and get a faster algorithm that constructs an identifier-independent network decomposition. In particular, we get an algorithm that computes a $(\log n, \log^2 n)$ decomposition in $O(\log^5 n + \log^4 n \log^* b)$ rounds  in the setting with $b$-bit identifiers and using $b$-bit messages.

\begin{theorem}[Informal Version of \cref{thm:log5balanced}]
\label{thm:log5balanced_informal}
There is a deterministic distributed algorithm, in the \congest model, that computes a $(\log n, \log^2 n)$ network decomposition in $O(\log^5 n + \log^4 n \log^* b)$ rounds in the setting with $b$-bit identifiers and using $b$-bit messages.
\end{theorem}

This leads to improvements for randomized algorithms in the \congest-model, in the shattering framework. For instance, for MIS, we get this result:

\begin{restatable}{corollary}{randomizedmis}
\label{thm:mis_randomized}
There is a randomized distributed algorithm that computes a maximal independent set in $O(\log \Delta \cdot \log\log n + \log^6 \log n)$ rounds of the \congest model, with high probability.
\end{restatable}

In contrast, the previous best algorithm had complexity $O(\log \Delta \cdot \log\log n) + 2^{O(\sqrt{\log\log n})}$~\cite{ghaffariPortmann2019}.
We get a similar result for $\Delta + 1$ coloring. 
\begin{restatable}{corollary}{randomizedcoloring}
\label{thm:coloring-in-shattering}
There is a randomized distributed algorithm, in the \congest model, that computes a $\Delta+1$ coloring in any $n$-node graph with maximum degree at most $\Delta$ in $O(\log \Delta + \log^6\log n)$ rounds, with high probability.
\end{restatable}

\subsection{Other Related Work}
\label{subsec:OtherRelated}
Here, we discuss some of the other related work that were not mentioned before. 

\paragraph{Usages of Decompositions} Network decomposition has been a central algorithm tool in distributed graph algorithms, since the work of Awerbuch et al. \cite{awerbuch89}. The work of \cite{ghaffari2017complexity, ghaffari2018derandomizing} generalized this much further: (1) \cite{ghaffari2017complexity} showed that one can use algorithms for $(C, D)$ decomposition to transform any sequential local algorithm (formally, in the $\mathsf{SLOCAL}$ model defined by \cite{ghaffari2017complexity}) to the \local model with only a slow down proportional to $CD$, when using a $(C, D)$ decomposition algorithm. (2) the work of \cite{ghaffari2018derandomizing} showed, how using the former together with the method of conditional expectation, one  can derandomize any $T$-round randomized \local model algorithm for any problem whose solution can be checked deterministically in $R$ rounds to a deterministic \local model algorithm with round complexity $O(CD(R+T))$ plus the time necessary to construct the network decomposition. 
Because of this, and the recent network decomposition algorithm of Rozho\v{n} and Ghaffari\cite{RozhonG19}, there is now a general efficient derandomization theorem for the \local model, which states that any $\poly(\log n)$-round randomized \local model algorithm for any locally checkable problem can be transformed to a deterministic \local model algorithm for the same problem, with only a $\poly(\log n)$ round slow down. With our improved network decomposition, the slow down is now improved  to $O(\log^5 n)$.
 
\paragraph{Decomposition Construction, Randomized Algorithms} Linial and Saks~\cite{linial93} gave a randomized algorithm that computes a $(\log n, \log n)$ weak-diameter network decomposition in $O(\log^2 n)$ rounds of the \congest model, with high probability. Elkin and Neiman~\cite{elkin16_decomp} presented a randomized algorithm that computes a $(\log n, \log n)$ strong-diameter network decomposition in $O(\log^2 n)$ rounds of the \congest model, with high probability.   
 
\paragraph{Decomposition Construction, Other Deterministic Results} Let us also mention some other deterministic results on constructing decompositions. As discussed before, the classic deterministic algorithm of Panconesi and Srinivasan provided a $(C, D)$ decomposition in $T$ rounds of the \local model where $C=D=T=2^{O(\sqrt{\log n})}$. Awerbuch et al.~\cite{awerbuch96} showed that, in the \local model, one can turn this into a $(\log n , \log n)$ decomposition in $2^{O(\sqrt{\log n})}$ rounds. Ghaffari~\cite{ghaffari2019MIS} gave a network decomposition algorithm matching the $C=D=T=2^{O(\sqrt{\log n})}$ bounds of Panconesi and Srinivasan in the \congest model. Ghaffari and Portmann \cite{ghaffariPortmann2019} gave an extension of this to power graphs $G^{k}$: in $k 2^{O(\sqrt{\log n})}$ rounds of the \congest model, their algorithm creates clusters colored with $2^{O(\sqrt{\log n})}$ colors, such that clusters of the same color have distance at least $k$, and each cluster has diameter at most $k 2^{O(\sqrt{\log n})}$ in graph $G$. They also discussed the applications of this power-graph decomposition for various problems including MIS, spanners, dominating set approximation, and neighborhood covers. 
The bounds were improved considerably in the work of Rozho\v{n} and Ghaffari \cite{RozhonG19}: in $k O(\log^8 n)$ rounds of the \congest model, their algorithm creates clusters colored with $O(\log n)$ colors, such that clusters of the same color have distance at least $k$, and each cluster has diameter  $O(k \log^3 n)$ in graph $G$. 

% \mtodo{please remind me of all related work that should be mentioned here, chang-pettie shattering in coloring}
% \newpage
\section{Faster Network Decomposition}
\label{sec:log5}

In this section we state our first technical contribution, a faster network decomposition algorithm. 

\begin{theorem}
\label{thm:log5}
Let $G$ be a graph on $n$ nodes where each node has a unique $b = O(\log n)$-bit identifier. 
There is a deterministic distributed algorithm that computes a network decomposition of $G$ with $O(\log n)$ colors and weak-diameter $O(\log^2 n)$, in $O(\log^5 n)$ rounds of the \congest   model with $\Theta(\log n)$ sized messages. 

Moreover, for each cluster $\fC$ of vertices in the output network decomposition, we have a Steiner tree $T_\fC$ with radius $O(\log^2 n)$ in $G$, for which the set of terminal nodes is equal to $\fC$. 
Each vertex of $G$ is in $O(\log n)$ Steiner trees of any given color out of the $O(\log n)$ color classes.  
\end{theorem}

Our improvement of the decomposition result of \cite{RozhonG19} comes from the improvement of their \emph{ball carving} algorithm. 
That is, we get a faster $O(\log^4 n)$-round algorithm that clusters at least half of the yet unclustered vertices into non-adjacent clusters, each cluster having a weak diameter of $O(\log^2 n)$. We remark that there is a randomized ball carving algorithm that, in $O(\log n)$ rounds of the \congest model, clusters at least half of the vertices into non-adjacent clusters with $O(\log n)$ weak-diameter in $O(\log n)$ rounds~\cite{linial93}, and one can also achieve the same with strong-diameter~\cite{elkin16_decomp}. These directly lead to $(\log n, \log n)$ weak and strong diameter decompositions in these two papers~\cite{linial93, elkin16_decomp}, in $O(\log^2 n)$ rounds of the \congest model.

\begin{theorem}
\label[theorem]{thm:log4}{(Ball carving algorithm)}
Consider an arbitrary $n$-node graph $G$ where each node has a unique $b = O(\log n)$-bit identifier, together with a subset $S \subseteq V$ of living vertices. 
There is a deterministic distributed algorithm that in $O(\log^4 n)$ rounds of the \congest  model finds a subset $S' \subseteq S$ of living vertices, where $|S'| \ge |S| / 2$, such that the subgraph $G[S']$ induced by $S'$ is partitioned into non-adjacent disjoint clusters, each of weak-diameter $O(\log^2 n)$ in $G$.  

Moreover, for each cluster $\fC$ of vertices, we have a Steiner tree $T_\fC$ with radius $O(\log^2 n)$ in $G$ for which the set of terminal nodes is equal to $\fC$. 
Each vertex in $G$ is in $O(\log n)$ Steiner trees.  
\end{theorem}

\cref{thm:log5} is obtained by $\log n$ applications of \cref{thm:log4}, starting from $S = V$. 
For each iteration $j \in [1, \log n]$, the set $S'$ are exactly nodes of color $j$ in the network decomposition, and we continue to the next iteration by setting $S \leftarrow S \setminus S'$.
The rest of this section describes the \textit{distributed ball carving}
 algorithm that proves \Cref{thm:log4}. 
\subsection{Intuition}
\label{sec:algorithm_intuition}

Our algorithm builds on the algorithm of Rozho\v{n} and Ghaffari \cite{RozhonG19}. Thus, before proving \cref{thm:log4}, we start by reviewing their algorithm. Afterwards, we discuss where their  algorithm has room for improvement and how our algorithm makes use of that. 

\paragraph{A Recap of the Ball Carving Algorithm of Rozho\v{n} and Ghaffari}
The ball carving algorithm of Rozho\v{n} and Ghaffari \cite{RozhonG19} that produces the clusters of one color class runs in $O(\log^6 n)$ rounds of the \local model and the weak-diameter of each cluster is bounded by $O(\log^3 n)$. \Cref{thm:log4} improves these two bounds to $O(\log^4 n)$ and $O(\log^2 n)$, respectively. 
In the original algorithm, at each point in time, a node in $S$ is either living or dead. Once a node is dead, it remains dead. Each living vertex is part of some cluster at every point in time, where each cluster is simply some set of vertices that changes over time.  At the beginning of the algorithm, each node forms a singleton cluster and the ID of that cluster is simply the $b$-bit identifier of the node. Throughout the algorithm, new nodes might join a given cluster, whereas other nodes might leave the cluster in order to join different clusters or because they got killed. The ID of the cluster does not change throughout the algorithm. A cluster might also cease to exist if all of its nodes either got killed or decided to join a different cluster. After the algorithm terminates, at least half of the vertices in $S$ are still alive. Moreover, each cluster is the union of one or more connected components in the graph induced by all the alive vertices. That is, there are no two neighboring nodes that are contained in different clusters.  

The algorithm consists of $b$ phases. The following is a a crucial invariant of the algorithm: at the end of the $i$-th phase, two neighboring clusters have the lowest $i$ bits of their ID in common. 
To preserve this invariant at the end of the $i$-th phase, given that it holds at the end of the $i-1$-th phase, clusters are split into blue and red clusters during the $i$-th phase based on their $i$-th bit. That is, if the $i$-th bit of the identifier is equal to $1$, we refer to a cluster as a blue cluster and otherwise, that is if the $i$-th bit is equal to $0$, we refer to a cluster as a red cluster. 
During the $i$-th phase, blue clusters can only grow, whereas red clusters can only shrink. At the end of the $i$-th phase, no blue cluster is neighboring with a red cluster. This suffices to preserve the invariant. Each phase consists of multiple steps. In each step, each node contained in a red cluster simply remains in the red cluster if it is not neighboring with any node in a blue cluster. Otherwise, the node in the red cluster proposes to join an arbitrary neighboring blue cluster. Thus, each blue cluster receives a certain number of proposals from neighboring nodes in red clusters. If the total number of proposals is at least a $1/(2b)$-fraction of the size of the blue cluster, all the proposing nodes join the blue cluster. Otherwise, the blue cluster decides to kill all proposing nodes and thus the blue cluster is not neighboring with any red cluster. The total number of killed vertices in each of the $b$ phases is at most a $1/(2b)$-fraction of the total number of nodes in $S$. Hence, throughout all of the $b$ phases, at most half of the vertices get killed. Moreover, each time a blue cluster does not kill all the proposing red nodes, its size increases by a $(1 + 1/(2b))$-factor. Thus, after $j$ such steps, the size of the blue cluster is at least $(1 + 1/(2b))^j$. As the size of each cluster is trivially bounded by $n$, each blue cluster can grow for at most $O(b \log n)$ steps and hence all blue clusters get separated from neighboring red clusters in at most $O(b \log n)$ steps. Hence, each of the $b = O(\log n)$ phases consists of $O(b \log n ) = O(\log^2 n)$ steps. As the weak diameter of each cluster grows by at most $2$ in each step, this directly implies that the weak diameter of each cluster is bounded by $O(b \cdot b\log n) = O(\log^3 n)$. Every single step can be implemented in $O(\log^3 n)$ rounds of the \local model, resulting in an overall round complexity of $O(\log^6 n)$ in the \local model.

\paragraph{Improved Version}
Next, we discuss on an intuitive level our improved algorithm compared to the original algorithm of Rozho\v{n} and Ghaffari\cite{RozhonG19}. 
Let us start with their algorithm and simply reduce the number of steps in each phase of the algorithm from $O(b \log n )$ down to $O(b)$. What would be the issue? 
The problem is that then, at the end of the phase, there might still be blue clusters neighboring red clusters.
However, each such blue cluster would have grown by a $(1 + 1/(2b))$-factor for all of the $O(b)$ steps in the phase, resulting in a constant factor increase of the cluster size. In some sense, this can also be seen as progress, as a constant factor growth can happen at most $O(\log n)$ times, at least if we assume that a cluster never shrinks (which it can). 
Alas, even assuming shrinking does not happen, the crucial invariant that after the $i$-th phase, the IDs of two neighboring clusters agree on the $i$ least significant bits does not hold anymore. 

We need, hence, a refined invariant. First, at each point in time, a given cluster $\fC$ is in some \textit{level} $\lev(\fC)$ that ranges from $0$ to $b$. 
The level is measuring the progress of a cluster in disconnecting itself from the other clusters; importantly, it is an individual measure for each cluster, whereas in the previous algorithm, this progress was measured for all clusters globally, by enforcing that at the end of the $i$-th phase, all clusters agree on the $i$ least significant bits in their identifier. 
Our new invariant, whose full statement is deferred to \cref{sec:tree}, implies that the identifiers of two neighboring clusters $\fC$ and $\fC'$ agree in the $\min(\lev(\fC),\lev(\fC'))$ least significant bits. For the purpose of this explanatory section, we call this property the \emph{level invariant}. 

Note that if at the end of the algorithm, each cluster is in level $b$, there are no two neighboring clusters, as desired. 
Furthermore, we would recover the old invariant if we assume that the level of each cluster increases by exactly one in each phase. 
But not every cluster's level will increase in each phase. Instead, in a given phase, the level of a cluster either increases by one or some other progress property happens: the cluster significantly ``grows'' in terms of the number of vertices that joined the cluster.

\paragraph{Growing Rule and Preserving the New Invariant}
We now describe our new algorithm in more detail: it has $O(b+ \log n)$ phases, each consisting of $O(b+\log n)$ steps. In each step, some vertices are proposing to join new clusters, according to the following rule. 
Recall that in the previous algorithm \cite{RozhonG19}, in phase $i$, vertices of clusters with the $i$-th bit equal to $0$ were proposing to join neighboring clusters with the $i$-th bit equal to $1$. 
Similarly, in our algorithm, vertices contained in some cluster $\fC$ that are neighboring with a cluster $\fC'$ having the same level as $\fC$ are proposing to join $\fC'$ if the $(\lev(C)+1)$-th bit of the identifier of $\fC$ is $0$, while the respective bit in the identifier of $\fC'$ is $1$. 
However, there is one more rule: if a vertex of $\fC$ neighbors with a cluster having a strictly smaller level than $\fC$, it prefers to propose to one such neighboring cluster $\fC'$ having the smallest level among all such neighboring clusters. 
As in the previous algorithm, if a sufficient amount of nodes propose to $\fC$, it decides to accept all proposals, while if there are not enough proposals, it kills proposing vertices, ``stalls" until the end of the phase and at the end of the phase increases its level.

The rule that a smaller level cluster $\fC$ is ``eating'' its higher level neighbor $\fC'$ is to enforce our level invariant: we know that the two clusters $\fC$ and $\fC'$, with $\fC$ having a strictly smaller level, agree on their $\lev(\fC)$ least significant bits. This invariant can fail once $\fC$ decides to increment its level. Hence, to justify going to the next level, $\fC$ also deletes the boundary with all higher-level neighboring clusters.
The level invariant follows from this new rule. The formal proof (of a more general invariant) is postponed to \cref{sec:tree}. 

\paragraph{Bounding the Number of Growing Steps} A crucial step in the analysis of the previous algorithm \cite{RozhonG19} is to argue that in each phase, each cluster can grow for at most  $O(b \log n)$ steps by a multiplicative factor of $1 + \Theta(1/b)$; otherwise, the cluster necessarily contains all the vertices of the graph. 
In our case, the picture is more complicated, as each cluster is eating the boundary vertices of its higher-level neighbors, while it is simultaneously eaten by its lower-level neighboring clusters. 
The rule that a cluster $\fC$ grows if the number of newly joined vertices is large with respect to the current number of vertices in $\fC$ does not work anymore. 

To remedy this problem, in our algorithm, each cluster $\fC$ possesses a certain number of tokens at every point in time. 
Initially, each cluster has a single token. 
During the course of the algorithm, $\fC$ obtains one token for every node that joins it. 
However, \emph{$\fC$ does not lose a token when a node leaves the cluster}. 
Instead, $\fC$ only loses tokens if it decides to kill all nodes proposing to it. 
In that case, $\fC$ pays a certain number of tokens (to be described later) for every node it kills. 

Each cluster decides to accept all proposals if the number of proposing nodes is a $\Omega(1/(b+\log n))$ fraction of its current number of tokens. 
Otherwise, the cluster kills all the proposing nodes. 
The parameters are set in such a way that the following holds: whenever a cluster is growing during the whole phase, the number of tokens it possesses at least doubles. 
On the other hand, if a cluster advances to the next level during a phase, then the number of its tokens remains at least half of what it was before (cf. Invariant 1 in \cref{sec:algorithm}). Notice that unlike this number of tokens, the size of the cluster can drop arbitrarily. 
Either way, each cluster progresses during each phase in terms of the number of tokens it possesses or by advancing to the next level. 

The final ingredient is that each cluster can create at most $O(b+\log n)$ tokens by joining new clusters. This will be proven later on. It implies that all clusters finish, i.e. are in the highest level, after $O(b+\log n)$ phases (cf. \cref{prop:clusters_finished}). If that would not be the case, then the total number of tokens an unfinished cluster would possess would exceed the total number of tokens that could possibly be created throughout the algorithm, a contradiction. Moreover, one can also show that at most half of the vertices get killed during the algorithm (cf. \cref{lem:small_num_of_deleted_nodes}). 

\subsection{Our Distributed Ball Carving Algorithm}

\label{sec:algorithm}

In this section we explain our algorithm for \cref{thm:log4}. 
Its analysis follows in \cref{sec:analysis,sec:tree}. 
%\paragraph{Notation}
%Recall that $b$ stands for number of bits. 
%Recall $S\subseteq V(G)$ is the subset of vertices of $G$ that are living at the beginning of the algorithm. 

\paragraph{Construction outline} The construction has $2(b+\log n) = O(\log n)$ phases. Each phase has $28(b + \log n) = O(\log n)$ steps. 
Initially, all nodes of $G$ are \emph{living}, during the construction some living nodes \emph{die}. 
%$S_i'$ denotes the set of living vertices \todo{do we  need it?}
Each living node is part of exactly one cluster. 
Initially, there is one cluster $\fC_v$ for each vertex $v \in V(G)$ and we define the identifier $\id(\fC)$ of $\fC$ as the unique identifier of $v$ and use $\id_{i}(\fC)$ to denote the $i$-th least significant bit of $\id(\fC)$. 
From now on, we talk only about identifiers of clusters and do not think of vertices as having identifiers, though they will still use them for simple symmetry breaking tasks. 
Also, at the beginning, the Steiner tree $T_{\fC_v}$ of a cluster $\fC_v$ contains just one node, namely $v$ itself, as a terminal node. 
Clusters will grow or shrink during the iterations, while their Steiner trees collecting their vertices can only grow. 
When a cluster does not contain any nodes, it does not  participate in the algorithm any more. 
%At the end of the algorithm, there are no edges between (nonempty) clusters whose diameter . 

\paragraph{Parameters of each cluster}
Each cluster $\fC$ keeps two other parameters besides its identifier $\id(\fC)$ to make its decisions: its number of tokens $t(\fC)$ and its level $\lev(\fC)$. 
The number of tokens can change in each step -- more precisely it is incremented by one whenever a new vertex joins $\fC$, while it does not decrease when a vertex leaves $\fC$. 
The number of tokens only decreases when $\fC$ actively deletes nodes. 
We define $t_i(\fC)$ as the number of tokens of $\fC$ at the beginning of the $i$-th phase and set $t_1(\fC) = 1$. 

Each cluster starts in level $0$. 
The level of each cluster does not change within a phase $i$ and can only increment by one between two phases; it is bounded by $b$. We denote with $\lev_i(\fC)$ the level of $\fC$ during phase $i$. 
Moreover, for the purpose of the analysis, we keep track of the potential $\Phi(\fC)$ of a cluster $\fC$ defined as $\Phi_i(\fC) = 3i - 2\lev_i(\fC) + \id_{\lev_i(\fC)+1}(\fC)$. 
The potential of each cluster stays the same within a phase. 

%At the beginning of the first phase, we have $i = 0$, and $t_0(\fC) = 1, \lev_0(\fC) = 0$ and $\Phi_0(\fC) =  \id_{1}(\fC)$. 

\paragraph{Description of a step}
In each step, first, each node $v$ of each cluster $\fC$ checks whether it is adjacent to a cluster $\fC'$ such that $\lev(\fC') < \lev(\fC)$. If so, then $v$ proposes to an arbitrary neighboring cluster $\fC'$ among the neighbors with the smallest level $\lev(\fC')$ and if there is a choice, it prefers to join clusters with $\id_{\lev(\fC') + 1}(\fC') = 1$. 
Otherwise, if there is a neighboring cluster $\fC'$ with $\lev(\fC') = \lev(\fC)$ and  $\id_{\lev(\fC') + 1}(\fC') = 1$, while  $\id_{\lev(\fC) + 1}(\fC) = 0$, then $v$ proposes to arbitrary such cluster. 

Second, each cluster $\fC$ collects the number of proposals it received. 
%This is done in $O(\log^2 n)$ rounds in the \local model and later we explain how we do it in the same number of rounds in the \congest model. 
Once the cluster has collected the number of proposals, it does the following. 
If there are $p$ proposing nodes, then they join $\fC$ if and only if 
$p \ge t(\fC) / (28 (b + \log n))$. The denominator is equal to the number of steps. 
If $\fC$ accepts these proposals, then $\fC$ receives $p$ new tokens, one from each newly joined node. 
On the other hand, if $\fC$ does not accept the proposals as their number is not sufficiently large, then $\fC$ decides to kill all those proposing nodes. These nodes are then removed from $G$. Cluster $\fC$ pays $p \cdot 14 (b + \log n)$ tokens for this, i.e., it pays $14 (b + \log n)$ tokens for every vertex that it deletes. 
These tokens are forever gone. 
Then the cluster does not participate in growing anymore, until the end of the phase and throughout that time we call that cluster \emph{stalling}. 
The cluster tells that it is stalling to neighboring nodes so that they do not propose to it. 
At the end of the phase, each stalling cluster increments its level by one. 

If the cluster is in level $b-1$ and goes to the last level $b$, it will not grow anymore during the whole algorithm, and we say that it has \emph{finished}. Other neighboring clusters can still eat its vertices (by this we mean that vertices of the finished clusters may still propose to join other clusters). 
%Note that in one step some nodes  of $\fC$ may join other clusters. However, this does \emph{not} decrease the number of tokens of $\fC$. This number decreases only when $\fC$ pays for deleting nodes. 
%

Whenever a node $u$ joins a cluster $\fC$ via a vertex $v \in \fC$, we add $u$ to the Steiner tree $T_\fC$ as a new terminal node and connect it via an edge $uv$. Whenever a node $u \in \fC$ is deleted or eaten by a different cluster, it stays in the Steiner tree $T_\fC$, but it is changed to a non-terminal node.

\paragraph{Construction invariants}
The construction is such that it preserves the following two invariants, as we formally prove in the next subsection.
\begin{enumerate}
    %\item Invariant 1: for each cluster $\fC$, we have that if it neighbors with $\fC'$ then they agree on $ \min(\lev(\fC), \lev(\fC'))$ rightmost bits. 
    \item Invariant 1: At the beginning of each phase $i$, we have $t_i(\fC) \ge 2^{i - 2\lev_i(\fC) - 1}$ unless $\fC$ is finished. 
    \item Invariant 2:  Whenever a node $u$ changes its cluster during some step in phase $i$, say it goes from $\fC$ to $\fC'$, it is the case that $\Phi_i(\fC') > \Phi_i(\fC)$. 
    Whenever we go to the next phase, the potential $\Phi(\fC)$ of each cluster does not decrease, i.e., $\Phi_{i+1}(\fC) \ge \Phi_i(\fC)$. 
\end{enumerate}

\subsection{Proving the Two Invariants}
\label{sec:analysis}
In this subsection, we prove Invariants 1 and 2 and that they imply that our algorithm outputs clusters of $O(\log^2 n)$ weak-diameter, while deleting at most $1/2$ fraction of vertices. 
The important fact that the resulting clusters do not neighbor as well as the fact that Steiner trees are indeed trees are postponed to \cref{sec:tree}, since their proofs require additional definitions. 

\begin{proposition}
\label[proposition]{lem:num_of_tokens_increases}
Invariant 1 is satisfied. That is, at the beginning of phase $i$, the current number of tokens $t_i(\fC)$ satisfies $t_i(\fC) \ge 2^{i - 2\lev_{i}(\fC) - 1}$, unless cluster $\fC$ is finished.
\end{proposition}
\begin{proof}
At the beginning of phase $1$, we have $\lev_1(\fC) = 0$ and $t_1(\fC) = 1$, hence Invariant 1 is satisfied. 
Now fix a phase $i$ and a cluster $\fC$ that is not finished at the end of the $i$-th phase. 
If the cluster decided to go to the next level during this phase, we have at the beginning of the phase $i+1$ that 
$
\lev_{i+1}(\fC) = \lev_{i}(\fC) + 1
$ 
 and, moreover, for the number of tokens $t_{i}(\fC)$, we have 
\[
t_{i+1}(\fC) 
\ge t_{i}(\fC) - \left( |t_{i}(\fC)|/ (28 (b + \log n)) \right) \cdot (14(b+\log n)) 
= t_{i}(\fC)/2, 
\]
because a given cluster can delete its boundary at most once in a given phase. 
Hence, by induction,
\[
t_{i+1}(\fC) \ge t_{i}(\fC)/2 
\geq  2^{i - 2\lev_{i}(\fC) - 1} /2
= 2^{i - 2\lev_{i+1}(\fC) + 2 - 1}/2
= 2^{(i+1) - 2\lev_{i+1}(\fC) -1}.
\]

Otherwise, we know that $\lev_{i+1}(\fC) = \lev_{i}(\fC)$ and $\fC$ was growing for all of the $28(b+\log n)$ steps of phase $i$. 
Hence, the number of tokens $t_{i}(\fC)$ at the beginning of phase $i+1$ satisfies
\[
t_{i+1}(\fC) \ge \left( 1 + 1/(28(b+\log n)) \right)^{28(b+\log n)} t_{i}(\fC) 
\ge 2t_{i}(\fC).
\]
This implies by the induction hypothesis that
\[
t_{i+1}(\fC) 
\ge 2 \cdot t_{i}(\fC) 
= 2 \cdot 2^{i - 2\lev_{i}(\fC) - 1 }
= 2^{(i+1) - 2\lev_{i}(\fC) - 1}.\qedhere
\]

\end{proof}

\begin{proposition}
\label[proposition]{lem:potential_increases}
Invariant 2 is satisfied. That is, whenever node $u$ changes its cluster during some step, say goes from $\fC$ to $\fC'$, it is the case that $\Phi_i(\fC') > \Phi_i(\fC)$. Moreover, whenever we go to the next phase, we have $\Phi_{i+1}(\fC) \ge \Phi_i(\fC)$. \end{proposition}
\begin{proof}
If $u$ goes from cluster $\fC$ to some cluster $\fC'$, then it is either because $\lev_i(\fC') < \lev_i(\fC)$, or because $\lev_i(\fC') = \lev_i(\fC)$ and $\id_{\lev_i(\fC) + 1}(\fC) = 0$ while $\id_{\lev_i(\fC') + 1}(\fC') = 1$. 
In the first case, 
\begin{align*}
    \Phi_i(\fC') 
    = 3i - 2\lev_i(\fC') + \id_{\lev_i(\fC') + 1}(\fC) 
    \ge 3i - 2(\lev_i(\fC) - 1) + \id_{\lev_i(\fC') + 1}(\fC)
    > 3i - 2\lev_i(\fC) . 
\end{align*}
In the second case,
\begin{align*}
    \Phi_i(\fC') 
    = 3i - 2\lev_i(\fC') + \id_{\lev_i(\fC') + 1}(\fC') 
    > 3i - 2\lev_i(\fC) + \id_{\lev_i(\fC) + 1}(\fC).
\end{align*}
Whenever we go from phase $i$ to phase $i+1$, we have
\begin{align*}
\Phi_{i+1}(\fC) 
&= 3(i+1) - 2\lev_{i+1}(\fC) + \id_{\lev_i(\fC) + 2}(\fC)
\ge 3i + 3 - 2(\lev_i(\fC) + 1) \\&
\ge 3i - 2\lev_i(\fC) + \id_{\lev_i(\fC) + 1}(\fC) = \Phi_i(\fC).\qedhere
\end{align*}
\end{proof}

\begin{proposition}
\label[proposition]{lem:node_in_few_clusters}
Each node can change its cluster at most $6(b + \log n) + 1$ times. 
\end{proposition}
\begin{proof}
At the beginning of phase $1$ of the algorithm each node $u$ in a cluster $\fC$ has $\Phi_1(\fC) \ge 0$. 
On the other hand, during any phase $i$, if $u \in \fC$, then $$\Phi_i(\fC) := 3i - 2\lev_i(\fC) + \id_{\lev_i(\fC) + 1}(\fC) \le 3i + 1.$$
Since the number of phases is equal to $2(b + \log n)$, we have $
\Phi_i(\fC) \le 6(b + \log n) + 1. 
$
Then, due to Invariant 2 (\cref{lem:potential_increases}), this means that $u$ changed its cluster at most  $6(b + \log n) + 1$ times, as whenever it changed its cluster, it went from $\fC$ to $\fC'$ such that $\fC'$ satisfies $\Phi_i(\fC') > \Phi_i(\fC)$ and when a new phase starts, we have for all clusters $\fC$ that $\Phi_{i+1}(\fC) \ge \Phi_i(\fC)$. 
\end{proof}

\begin{proposition}
\label[proposition]{lem:bound_overall_num_of_tokens}
The total number of tokens generated by nodes throughout the algorithm is at most $7|S| (b + \log n)$. 
\end{proposition}
\begin{proof}
Each node generates a token at the very beginning of the algorithm and then it generates one token whenever it changes its cluster. 
By \cref{lem:node_in_few_clusters}, each node can generate at most $6(b + \log n) + 1$ tokens by changing a cluster. 
Hence, the total number of tokens generated is at most $|S| (6(b + \log n) + 2) \le 7|S| (b + \log n)$. 
\end{proof}

\begin{proposition}
\label[proposition]{lem:small_num_of_deleted_nodes}
In the end, the number of deleted vertices is at most $|S| / 2$. 
\end{proposition}
\begin{proof}
Whenever a node is deleted from $S$, we permanently set aside $14 (b + \log n)$ tokens. Hence, by \cref{lem:bound_overall_num_of_tokens}, the total number of nodes deleted is at most
\[
\frac{7|S| (b + \log n)}{14 (b + \log n)} 
= |S| / 2.\qedhere
\]
\end{proof}

\begin{proposition}
\label[proposition]{lem:small_weak_diameter}
Per step, the diameter of every Steiner tree $T_\fC$ grows additively by at most $2$. Hence, in the end of the algorithm, the diameter of each graph $T_\fC$ and, therefore, the weak-diameter of each $\fC$, is bounded by $O(\log^2(n))$. 
Moreover, each vertex of $G$ is in at most $O(\log n)$ different Steiner trees $T_\fC$. 
\end{proposition}
\begin{proof}
In one step of a phase, we increase the Steiner tree $T_\fC$ only by adding new leaves to it (though the fact that each vertex is added to $T_\fC$ at most once and hence it is a tree is proved only in \cref{prop:trees_are_trees}). 
We have $O(\log n)$ phases and each phase has $O(\log n)$ steps, hence the diameter of each $T_\fC$ is bounded by $O(\log^2 n)$, in the end. 
The last part follows from the fact that whenever a vertex $u$ is added to a new Steiner tree, $u$ changes its cluster. 
This can happen at most $6(b + \log n) + 1 = O(\log n)$ times, by \cref{lem:node_in_few_clusters}. 
\end{proof}

\begin{proposition}
\label[proposition]{prop:clusters_finished}
At the end of phase $i_{last} = 2(b + \log n)$, the level of each cluster $\fC$ is equal to $\lev_{i_{last}}(\fC) = b$, i.e., $\fC$ is finished. 
%Moreover, there are no edges between clusters. 
\end{proposition}
\begin{proof}
The first part follows from Invariant 1 (\cref{lem:num_of_tokens_increases}) as follows. 
Unless $\fC$ is finished, Invariant 1 maintains that $t_i(\fC) \ge 2^{i - 2 \lev_i(\fC) - 1}$. This means that if $\fC$ is still not finished at the end of the phase $i_{last} = 2(b + \log n)$, then we would have 
\[
t_{i_{last}}(\fC)
\ge 2^{{i_{last}} - 2 \lev_{i_{last}}(\fC) - 1}
\ge 2^{2(b + \log n) - 2 b - 1}
\ge n^2/2, 
\]
a contradiction with \cref{lem:bound_overall_num_of_tokens}. 
%
%Now, by Invariant 1 / \cref{lem:agree_on_rightmost_bits}, if there are two neighboring clusters $\fC$ and $\fC'$, they have to agree on their $\min(\lev_{i_{last}}(\fC), \lev_{i_{last}}(\fC')) = b$ bits. 
%But this is a contradiction with the fact that the identifiers of clusters are unique. 
\end{proof}

\subsection{Transcript Tree and Isolating Clusters}
\label{sec:tree}

In this subsection, we show that the final clustering produced by the algorithm described in \cref{sec:algorithm} satisfies that there are no two neighboring clusters. This is stated as the following proposition. 

\begin{proposition}
\label{prop:clusters_nonadjacent}
At the end of the algorithm, resulting clusters are nonadjacent. 
\end{proposition}
That is, once the algorithm terminates, there does not exist an edge with both endpoints being alive and contained in different clusters. 
We also prove the following fact. 
\begin{proposition}
\label{prop:trees_are_trees}
Each vertex $v$ is added at most once to each $T_\fC$, hence, the graphs $T_\fC$ are trees. 
\end{proposition}
To that end, we define an invariant that holds throughout the execution of the algorithm and which implies the properties stated above.
To define the invariant, we consider a fixed $4(b+\log n)$-ary rooted tree (i.e., the branching factor is twice the number of phases) of depth $b$ called \emph{the transcript tree $T$}\footnote{Try saying it three times in a row. }, where the root is defined to have depth $0$. Throughout the course of the algorithm, we map each non-empty cluster to one of the nodes in the tree $T$ by a mapping $\pi$. 
At the beginning, each cluster simply maps to the root of $T$. 
A cluster only changes the node it maps to when its level is increased, using the following rule. 
If a cluster $\fC$ advances from level $\lev(\fC)$ to level $\lev(\fC) + 1$ between phases $i$ and $i+1$, it is remapped to the  $(2i + \id_{lev(\fC) + 1}(\fC))$-th child of the node it previously mapped to. 
Notice that for each non-root node of $T$, there is only one phase when new clusters can be mapped to it (if the node is the $(2i)$-th or $(2i+1)$-th child, it is phase $i$). 
From that time on, unless the node is a leaf node of $T$, the clusters are gradually reassigned to its children or completely deleted from $T$ if they become empty. 
Notice that the current level of each cluster is equal to the depth of the node that this cluster currently maps to. Finally, our construction satisfies the following two properties:

\begin{observation}
\label[observation]{obs:tree_agree_on_bits}
The identifiers of all clusters that map to a given node at depth $d$ agree on the $d$ least significant bits.
\end{observation}

\begin{proposition}
\label[proposition]{prop:tree_stalling}
Suppose that $\fC$ is a stalling cluster. 
Then it does not  neighbor with higher level clusters and if $\id_{\lev_i(\fC)+1}(\fC) = 1$, it does not neighbor with any cluster $\fC'$ of the same level with $\id_{\lev_i(\fC')+1}(\fC') = 0$. 
\end{proposition}
\begin{proof}
Whenever a cluster $\fC$ deletes its boundary and starts stalling, each neighboring node $u$ that considered proposing to $\fC$, but did not, either proposed to a cluster of level strictly smaller than $\lev(\fC)$, or it proposed to a cluster $\fC'$ in the same level, but then $\id_{\lev(\fC')+1}(\fC') \ge \id_{\lev(\fC)+1}(\fC)$. 
Then, $u$ is either deleted, or it joins $\fC'$. 
So, a cluster $\fC$ that starts stalling can be neighboring with another cluster $\fC'$, but then the level of $\fC'$ is either strictly smaller, or it is the same, but $\id_{\lev(\fC')+1}(\fC') \ge \id_{\lev(\fC)+1}(\fC)$. 

In the following steps, a node of $\fC$ can be eaten by one of the neighboring clusters, but this does not create new neighbors of $\fC$, or a connection with a different cluster $\fC''$ is created by that cluster eating a node of some neighboring cluster $\fC'$. 
However, $\fC''$ is either of smaller level than $\fC'$, or it is the same level, but with $\id_{\lev(\fC'')+1}(\fC'') \ge \id_{\lev(\fC')+1}(\fC')$. Hence, this new connection is still allowed. 
\end{proof}

We now prove that the algorithm described in \cref{sec:algorithm} satisfies the following crucial invariant throughout the course of the algorithm. \cref{fig:tree} might help to obtain a better intuition.

\begin{proposition}
\label[proposition]{prop:tree_invariant}
 Whenever two clusters $\fC$ and $\fC'$ are neighboring, then either $\pi(\fC)$ is an ancestor of $\pi(\fC')$---i.e., $\fC$ is mapped to a node that lies on the unique path between the node $\fC'$ maps to and the root in $T$---or $\pi(\fC')$ is an ancestor of $\pi(\fC)$.
\end{proposition}

\begin{figure}
    \centering
    \includegraphics[width=\textwidth]{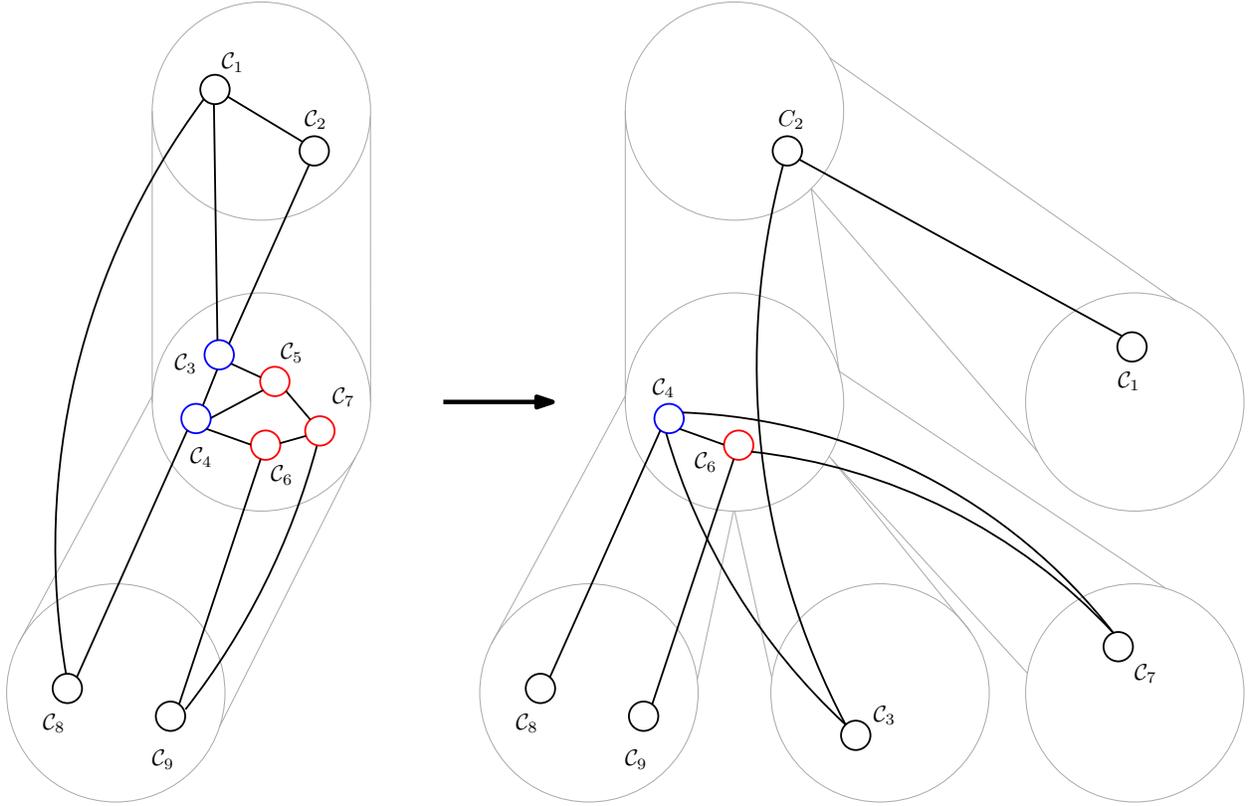}
    \caption{The figure captures a possible change in cluster mapping between the beginning of phase $i$ and the beginning of phase $i+1$ of the algorithm, with focus on one node of the transcript tree $T$ in depth $d$ containing clusters $\fC_3, \fC_4, \fC_5, \fC_6, \fC_7$ that are colored blue if their $(d+1)$'th bit is equal to $1$ and red otherwise (i.e., red vertices are proposing to blue clusters). Two clusters are connected by an edge in the figure if they are neighboring. \\
    The cluster $\fC_1$ is eating clusters $\fC_3$ and $\fC_8$ (by this we mean their vertices propose to $\fC_1$) and during the phase, it decides to delete its boundary with $\fC_3$ and $\fC_8$ and to go to the next level $d$ at the end of the phase -- it is reassigned to a node of $T$ in depth $d$. 
    The cluster $\fC_2$ is eating $\fC_3$ during the whole phase and it will continue eating it even in the next phase. 
    The cluster $\fC_3$ is eating $\fC_5$, until it decides to delete its boundary with it and to go to the next level $d+1$. 
    The cluster $\fC_4$ is eating $\fC_5, \fC_6, \fC_8$ and later in the phase also $\fC_7$. 
    All vertices of $\fC_5$ leave that cluster at some point during this phase so the whole cluster is dissolved and we do not map it to $T$ any more. 
    The cluster $\fC_6$ is eating $\fC_9$ and later in the phase also $\fC_7$, after $\fC_7$ decides to delete its boundary to $\fC_9$. 
    }
    \label{fig:tree}
\end{figure}

\begin{proof}%[Proof of \cref{prop:tree_invariant}]
We prove \cref{prop:tree_invariant} by induction on the number of executed steps of the algorithm.
We prove that it stays satisfied after every step of the algorithm, and also between any two phases, when stalling clusters go to the next level. 
We note that the property to prove holds at the beginning of the algorithm, since all the clusters are mapped to the root node of $T$.

Next, fix a step $j$ of some phase $i$ and assume that the property to prove is satisfied right at the beginning of the step. We now consider some arbitrary edge $\{u,v\}$, where both $u$ and $v$ have not been deleted. In order to prove that the invariant holds after step $j$, it suffices to show that after step $j$, nodes $u$ and $v$ are not contained in two different clusters such that none of the two clusters is an ancestor of the other cluster.

This holds because whenever $u \in \fC_u$ or $v\in \fC_v$, respectively, proposes to some cluster $\fC'_u$ or $\fC'_v$, respectively, by the induction hypothesis, $\pi(\fC'_u)$ is an ancestor of $\pi(\fC_u)$ (possibly, $\pi(\fC'_u) = \pi(\fC_u)$) and similarly we have that $\pi(\fC'_v)$ is an ancestor of $\pi(\fC_v)$. 
By the induction hypothesis, we also know that either $\pi(\fC_u)$ is an ancestor of $\pi(\fC_v)$, or the other way around. 
Putting these facts together, we get that either $\pi(\fC'_u)$ is the ancestor of $\pi(\fC'_v)$, or the other way around, as desired. 

Second, we show that the property stays satisfied between two phases $i$ and $i+1$. 
We again consider an arbitrary edge $\{u, v\}$ with $u \in \fC_u$ and $v \in \fC_v$. 
If neither $u$ nor $v$ stalled, there is nothing to prove. If both $\fC_u$ and $\fC_v$ stalled, by \cref{prop:tree_stalling}, we have $\lev_i(\fC_u) = \lev_i(\fC_v) = \lev$ and $\id_{\lev+1}(\fC_u) = \id_{\lev+1}(\fC_v)$. By the induction hypothesis, $\pi(\fC_u) = \pi(\fC_v)$, hence both $\fC_u$ and $\fC_v$ are remapped to the same node of the transcript tree $T$ between the two phases. 
If $u$ stalled but $v$ did not, by \cref{prop:tree_stalling} and the induction hypothesis, $\pi(\fC_v)$ is an ancestor of $\pi(\fC_u)$. Hence, after remapping $\fC_u$ to one of the children of the node it previously mapped to, the induction hypothesis is still satisfied. 
\end{proof}

Now, we are ready to prove \cref{prop:clusters_nonadjacent} and \cref{prop:trees_are_trees}. 
\begin{proof}[Proof of \cref{prop:clusters_nonadjacent}]
By \cref{prop:clusters_finished}, at the end of the algorithm, all resulting clusters are in level $b$. 
Hence, by \cref{prop:tree_invariant}, two adjacent clusters need to map to the same node of $T$ at depth $b$. 
However, by \cref{obs:tree_agree_on_bits}, the two clusters then agree on their identifiers, which is a contradiction with their uniqueness. 
\end{proof}

\begin{proof}[Proof of \cref{prop:trees_are_trees}]
Fix some $T_\fC$ and a vertex $u$ that was added to $\fC$ at some point during the algorithm. 
Suppose $u$ leaves $\fC$ and joins some cluster $\fC'$. We prove that $u$ cannot join $\fC$ in the future. 
First, suppose $\fC'$ is currently in strictly smaller level than $\fC$. Then we claim $u$ cannot join a cluster from the subtree of $\pi(\fC)$, and $\fC$ in particular, anymore. 
This is because clusters cannot be remapped to $\pi(\fC)$ anymore and clusters from the subtree of $\pi(\fC)$ do not have any connections to other clusters beside clusters in the path from $\pi(\fC)$ to the root, by \cref{prop:tree_invariant}. But vertices in those clusters never propose to clusters in the subtree of $\pi(\fC)$, since they have a smaller level. 

Similarly, if $u$ leaves $\fC$ and joins a cluster $\fC'$ that is currently in the same level $d$, by \cref{prop:tree_invariant} we have $\pi(\fC) = \pi(\fC')$ and $\id_{d+1}(\fC) = 0$ while $\id_{d+1}(\fC') = 1$. 
Whenever $u$ is later eaten by a cluster with strictly smaller level than $d$ or $\fC$ goes to the next level, we argue as in the previous case. 
Otherwise, after $\fC'$ deletes its boundary to $\fC$ and starts stalling, we have that $\fC'$ cannot become adjacent to $\fC$ during this phase and this holds also  during next phases, since, by induction, $\fC$ can eat only vertices from some other branches of the subtree of $\pi(\fC)$ than the branch of $\pi(\fC')$ and clusters in those branches are not adjacent to $\pi(\fC')$ by \cref{prop:tree_invariant}. 
Hence, each vertex is added to $T_\fC$ at most once and $T_\fC$ is a tree. 
\end{proof}

\subsection{Wrapping up}

We are now ready to wrap up the analysis of our distributed ball carving algorithm and present the proof of \Cref{thm:log4}.
\begin{proof}[Proof of \Cref{thm:log4} ]
The total number of deleted nodes is at most $|S|/2$ by \cref{lem:small_num_of_deleted_nodes}. 
The fact that the resulting clusters are not neighboring follows from \cref{prop:clusters_nonadjacent}.  
%The clusters are not neighboring because of \cref{prop:clusters_finished}. 
The corresponding Steiner trees are trees via \cref{prop:trees_are_trees}, have weak-diameter $O(\log^2 n)$ and each edge is in at most $O(\log n)$ Steiner trees by \cref{lem:small_weak_diameter}. 

Finally, we bound the running time. 
In the \local model, it is bounded by $O(\log^4 n)$, since the algorithm has $O(\log n)$ phases, each having $O(\log n)$ steps and each step can be implemented in the number of rounds proportional to the weak diameter of each cluster, which is bounded by $O(\log^2 n)$. 

In the \congest  model, we first verify that an $O(\log^5 n)$ upper bound holds because each step can be implemented in $O(\log^3 n)$ rounds as follows: 
First, every step starts by nodes proposing to join a neighboring cluster, provided there is a suitable one. This step is implemented in two \congest  model rounds. %: First, each node neighboring $u$ sends to $u$ its level and its identifier and then $u$ picks a potential node $v$ to propose to and in that case it tells this to $v$. 
Second, each root of the Steiner tree $\fC$ needs to collect how many nodes are proposing to the cluster. 
Since each edge is contained in $O(\log n)$ Steiner trees and the diameter of each Steiner tree is $O(\log^2 n)$, this can be done in $O(\log^3 n)$ steps. Finally, the cluster $\fC$ needs to decide whether it will grow or not and this information is then broadcasted via $T_\fC$ to all proposing nodes. This can again be done in $O(\log^3 n)$ rounds.  

Using \cref{cor:pipelining_with_overlap} from \cref{sec:aggregating}, we can speed up the aggregation of the summation and the broadcast in every cluster so that it runs in parallel for all the clusters in  $O(\log^2 n)$  rounds. This recovers the same round complexity of $O(\log^4 n)$ for the \congest model, matching that of the \local model. 
\end{proof}

\subsection{Example Applications: MIS and Coloring}
% \todo{add beyond}
\label{sec:mis}
As two prominent examples of applications, below we mention how we obtain $O(\log^5 n)$ round deterministic \congest model algorithms for the maximal independent set and $\Delta+1$ coloring problems. These improve on the $O(\log^7 n)$-round \local model and $O(\log^8 n)$-round \congest model algorithms of Rozho\v{n} and Ghaffari~\cite{RozhonG19}. We note that similar polynomial improvements happen for all other applications of network decompositon, many of which are discussed in~\cite{RozhonG19}.

\deterministicmis
\begin{proof}[Proof Sketch]
We process the color classes of the network decomposition, one by one. When processing clusters of color $i$, first, we remove each vertex that is adjacent to a node that is already in the MIS. Then, for each cluster, we run the deterministic MIS algorithm of Censor-Hillel et al.\cite{censor2017derandomizing}, which computes an MIS in $O(D\log^2 n)$ rounds in any $n$-node graph of diameter $D$. Since each cluster has weak diameter $O(\log^2 n)$, running this algorithm in one cluster would be doable in $O(\log^4 n)$ rounds. Running the algorithm for different clusters needs more care, as their Steiner trees are not edge disjoint: The MIS algorithm of Censor-Hillel et al.\cite{censor2017derandomizing} is based on derandomizing the $O(\log n)$ round algorithm of \cite{ghaffari2016MIS}. They observe that each round needs only pairwise independence, which thus means only $O(\log n)$ bits of randomness. Then, these bits are fixed one by one, using the method of conditional expectation. To perform this, the key step is to determine how to fix each single bit (conditioned on the bits fixed so far). For that, each node needs to compute (a certain pessimistic estimator of) the probability of it being in the MIS or neighboring an MIS node, under the two possibilities of the single randomness bit that we are examining. This is done via $1$ round of communication with the neighbors in the MIS problem, and that part we can easily do in our setting as the nodes of different clusters are disjoint (even though their Steiner trees are not). Then, the algorithm of Censor-Hillel et al.\cite{censor2017derandomizing} aggregates the sum of these probability estimators, using a convergecast on the global BFS tree of the network, with depth $D$, in $D$ rounds. To perform this part, we make each cluster use its Steiner tree. These Steiner trees are not disjoint, but fortunately, each vertex is in at most $O(\log n)$ Steiner trees. Hence, we can apply the pipelining of \Cref{cor:pipelining_with_overlap}, which allows us to aggregate the summations for different clusters in parallel, in $O(\log^2 n + \log n) = O(\log^2 n)$ rounds. Once these sums are gathered at the center, it can be decided how to fix this one bit of the randomness of this round of  \cite{ghaffari2016MIS}, and we can proceed to the next bit. There are $O(\log n)$ rounds and we need to fix $O(\log n)$ bits for each. Hence, overall, the round complexity of computing an MIS for each cluster of one color class, all at the same time, is $O(\log^4 n)$ rounds of the \congest model. This is the complexity for one color class of the decompositon. Since the decomposition has $O(\log n)$ colors, the overall complexity of solving MIS, given the network decomposition, is $O(\log^5 n)$. When put together with the $O(\log^5 n)$ round complexity needed for computing the decomposition via \cref{thm:log5}, we have a deterministic MIS algorithm that runs in $O(\log^5 n)$ rounds.
\end{proof}

\deterministiccoloring

% \todo{I believe that the right complexity is $O(\log^6 n)$.}
\begin{proof}
The proof is similar to the MIS result, with only one exception: when solving the problem in each cluster, instead of the \congest-model MIS algorithm of Censor-Hillel et al.~\cite{censor2017derandomizing}, we apply the \congest-model list-coloring algorithm of Bamberger et al.~\cite{bamberger2020efficient}.
\end{proof}

% \todo{corollary for coloring?}

\section{Identifier-Independent Network Decomposition}
\label{sec:handling_identifiers}

In this section, we explain how one can obtain a much milder dependence of the round complexity on the length of identifiers (and bit capacity of each edge) $b$. 
Specifically, the round complexity $\poly(b \cdot \log n)$ is improved to $(\log^* b) \cdot \poly(\log n)$. 
Note that in the \local model, i.e., without constraints on the capacity of edges, this is a direct implication of distance coloring (cf. Remark 2.10 in \cite{RozhonG19}). 

In standard, deterministic, applications, we have $b = \Theta(\log n)$, so we do not get an improvement over the previous formulation of the algorithm.
However, in the shattering framework, we have $N = O(\log n)$ and $b=\Theta(\log n)$, so we get an improved complexity from $\poly(\log n)$ down to $\poly(\log \log n)$. 

In this section, the idea of our improvement is explained by modifying the algorithm of \cite{RozhonG19} explained in \cref{sec:algorithm_intuition}. 
The complexity of their algorithm is $O(b^4 \log^3n)$ and we show how to change it to $O(\log^7 n + (\log^* b)\cdot \log^6 n)$. 
In \cref{sec:faster_construction_without_identifiers}, we improve the round-complexity of the algorithm from \cref{thm:log5} from $O(b^4 \log n)$ to $O(\log^5 n + (\log^*b)\cdot \log^4 n )$.
\subsection{Balanced Coloring}
\label{sec:balanced_coloring}
\begin{lemma}
\label{thm:balanced_coloring}
Consider a graph $G=(V, E)$ that has no isolated vertices and where each node has a $b$-bit identifier. There is a deterministic distributed algorithm in the \congest model that, in $O(\log^* b)$ rounds, colors the vertices of $V$ blue or red such that each color has at most $3|V|/4$ vertices. 
\end{lemma}
\begin{proof}
Let each node $v$ choose one of its edges in $G$ arbitrarily, and indicate this as an outgoing edge from $v$. Let $H$ be the spanning subgraph of $G$ defined by the set of all chosen edges. Call a vertex $u$ \textit{heavy} if its in-degree in $H$ is at least $10$, and \textit{light} otherwise. Since $H$ has at most $|V|$ outgoing edges, there are at most $|V|/10$ heavy vertices.
Let $H'$ be the subgraph of $H$ induced by light vertices. We handle vertices of $H'$ in two categories of isolated and non-isolated vertices. 

(A) Light vertices that are isolated in $H'$ must have their chosen outgoing edge connect to a heavy vertex. These outgoing edges define stars, at most one centered on each heavy vertex. Each heavy vertex computes a coloring of itself and all the isolated light edges that point to it, such that the number of colors in the star differ by at most $1$. This way, we have a discrepancy---i.e., the absolute difference in the number of nodes of the two colors---of at most $1$ in each star, and thus overall a discrepancy of at most $|V|/10$.

(B) Non-isolated vertices of $H'$ form a graph with minimum degree at least $1$ and maximum degree at most $11$. Compute a maximal independent set $S$ of $(H')^2$--- that is, the graph on vertices of $H'$ where we connect two of them if their distance is at most $2$ in $H'$---in $O(\log^* b)$ rounds, using Linial's classical algorithm\cite{linial1987LOCAL}. Then, each node of $H'$ that is not in $S$ chooses the closest node in $S$ as its cluster center. Since we have a maximal independent set of $(H')^2$, each node has a cluster center within distance $3$ in $H'$. Moreover, each cluster has at least two vertices, i.e., the cluster center and all of its neighbors, which is at least one neighbor. Each node in $S$ computes a coloring of the vertices of its own cluster, in a manner that the number of colors in the cluster differ by at most one. We have no cluster with a single vertex. Each cluster with $2$ vertices has no discrepancy and each cluster with $3$ or more vertices has discrepancy at most $1$. This means, the discrepancy in the coloring of $H'$ is at most $|V|/3$. 

Taking the discrepancies in the two parts into account, we have discrepancy at most $|V|(1/3+1/10) = 13|V|/30.$ Therefore, each color has at least $17/60|V| > |V|/4$ vertices.
\end{proof}

\begin{lemma}\label[lemma]{thm:balanced_coloring2}
Consider a cluster graph where no cluster is isolated, and each cluster has a unique $b$-bit identifier. Moreover, each cluster $\fC$ has a Steiner tree $T_\fC$ of diameter $R$, such that each node is in at most $O(\log n)$ of these Steiner trees. 
There is a deterministic distributed algorithm in the \congest model with $O(b)$-bit messages that, in $O((R + \log n) \cdot \log^* b)$ rounds, colors the clusters blue or red such that each color has at most a $3/4$ fraction of the clusters.
\end{lemma}
\begin{proof} We follow an approach similar to \Cref{thm:balanced_coloring}, but we have to deal with two issues: (1) nodes are replaced with clusters of weak-diameter $R$, (2) the Steiner trees of the clusters are not disjoint, and each node can be in up to $O(\log n)$ Steiner trees.

\paragraph{Selecting An Outgoing Edge Per Cluster} First, we select one outgoing edge for each cluster, in the cluster graph. For that, any two neighbors exchange their cluster identifier, in one round. Then, any node $w$ in a cluster $\fC$ that is neighboring some node $w'$ in another cluster $\fC'$ creates a proposed outgoing edge $\langle\fC'.ID, w.ID, w'.ID\rangle$. We then convergecast the minimum of these proposals to the root of the cluster $\fC$. We do this for all the clusters at the same time, in $O(R+\log n)$ rounds, using the pipeling of \Cref{cor:pipelining_with_overlap}. At the end, the center of $\fC$ knows the winning proposal $\langle\fC'.ID, w.ID, w'.ID\rangle$ that connects it to some other cluster $\fC'$. In this case, the outgoing edge in the cluster graph is $\fC\rightarrow \fC'$, and we consider the edge $w\rightarrow w'$ as the physical embodiment of this outgoing edge. By performing a broadcast in each cluster, and all clusters at the same time, we can inform all nodes of the cluster of the selected single outgoing edge, in $O(R+\log n)$ rounds, using the pipelining of \Cref{cor:pipelining_with_overlap}. In particular, node $w$ learns that its edge $\{w, w'\}$ is selected as the outgoing edge $w\rightarrow w'$ of its cluster. It can also inform $w'$ about this, in one additional round.

\paragraph{Identifying Light and Heavy Clusters} We call a cluster heavy if it has at least $10$ incoming edges, and light otherwise. Our next task is to inform each cluster whether it is heavy or light. For each cluster $\fC'$, each node $w' \in \fC'$ that has an incoming edge $w'\leftarrow w$ from another cluster starts a message describing this edge as $\langle\fC.ID, w'.ID, w.ID\rangle$. We then convergecast all of these incoming edge messages in each cluster, or at most $11$ of them, if there are more. This can be done for all clusters at the same time in $O(R+\log n)$ rounds, using the pipelining of \Cref{cor:pipelining_with_overlap}. At the end, each cluster center knows whether it has more than $11$ incoming edges or not, i.e., whether it is heavy or not. Moreover, every light node knows all of its incoming edges. Using one broadcast per cluster, by \cref{cor:pipelining_with_overlap}, we can also inform all nodes of the cluster whether the cluster is heavy or light, and about all the incoming edges if it is light, in $O(R+\log n)$ additional rounds.

\paragraph{Coloring Non-Isolated Light Clusters} 
Consider all the incoming and outgoing edges as undirected edges, and consider the subgraph $H$ made of light clusters who have at least one such edge. By repeating the above communication scheme, we can identify all such clusters and in fact implement one round of the \congest model on the graph $H'$, in $O(R+\log n)$ rounds of communication on the base graph. At this point, it is easy to follow the steps of \Cref{thm:balanced_coloring} to color light clusters of $H'$: we compute an MIS of $H^2$, in $O((R+\log n) \log^* b)$ rounds, and then each MIS cluster $\fC$ has to determine the red/blue colors of itself and its neighboring clusters. It does so in a way that the discrepancy between the number of red and blue colors that $\fC$ gives out is at most $1$.

\paragraph{Coloring Heavy Clusters, and their Incoming Isolated Light Clusters }
What is left is coloring each heavy cluster $\fC$, as well as all the light clusters isolated in $H'$ and whose selected outgoing edge was therefore to a heavy cluster. Each cluster $\fC$ does this on its own, for itself, and all such light clusters that have an outgoing edge to $\fC$. First, we initiate a token (carrying $O(1)$ bits), at the physical embodiment of every such incoming edge. We also start one token at the root of the heavy cluster. Then, we convergecast these tokens on the Steiner tree of $\fC$, in a synchronized manner from depth $R$ to the root. That is, we start with nodes of depth $R$, they send their tokens to nodes of depth $R-1$ in one round, they send their tokens to the nodes of depth $R-2$ in another round, and so on. We do this for all heavy clusters at the same time, in $O(R+\log n)$ rounds, by allocating $O(1)$ bits of the messages of each round to each of the Steiner trees that includes the edge. Notice that this is possible as we have $b=\Omega(\log n)$-bit messages and each node is in at most $O(\log n)$ trees.
Now, for each Steiner tree, every time that a node $v$ on this Steiner tree receives some tokens from its children, node $v$ pairs the tokens up with each other in pairs of two, except for leaving at most one token not paired if the number is odd. Tokens that are paired are sent backward along the same tree, from $v$ to the physical incoming edge that initiated the token. In each pair, one token carries color red and the other carries token blue. If the number of tokens that $v$ had received was odd, then it forwards the one remaining unpaired token to its parent in the Steiner tree, in the next round. If a token is left unpaired at the root, we color it arbitrarily. After performing this for $2R$ rounds, all tokens are paired up, with the exception of at most one token in the case their number is odd. 
Moreover, they have arrived back at the incoming endpoint of the physical incoming edge. Then, using one additional round we can send the color to the other endpoint of the physical incoming edge, and using another convergecast in each cluster, we can inform the center of each light cluster (that had no neighbor in $H'$) of the color that it received in this scheme, in $O(R+\log n)$ rounds, for all clusters at the same time, using the pipelining of \Cref{cor:pipelining_with_overlap}. This concludes the description of the procedure that implements the balanced coloring algorithm of \Cref{thm:balanced_coloring} on the clusters, in $O((R+\log n)\log^* b)$ rounds.
\end{proof}
\begin{remark}
Any deterministic \local-model algorithm for balanced coloring needs $\Omega(\log^* n)$ rounds, even on a cycle. 
\end{remark}
\begin{proof}
Suppose for the sake of contradiction that there is a deterministic algorithm $\mathcal{A}$ that on any $n$-node cycle with $O(\log n)$-bit identifiers, in $T\leq (\log^* n)/100$ rounds, computes a balanced coloring, such that at most $3/4$ of the nodes are blue and at most $3/4$ of them are red. Consider $n$ separate $n$-node graphs, where the $i^{th}$ one has identifiers in $[(i-1)\cdot n + 1, i\cdot n]$. By Linial's well-known lower bound~\cite{linial1987LOCAL}, we know that on each cycle, there is a configuration of the identifiers such that algorithm $\mathcal{A}$, when run on that cycle with those identifiers, colors some consecutive set of at least $H\geq (\log^* n)/5$ nodes on the cycle all blue, or all red. This is because, otherwise, we could then extend the coloring of $\mathcal{A}$ to a $4$-coloring, by processing each consecutive monochromatic path in time at most $H$ and computing a $2$-coloring of its vertices. This would result in a $4$-coloring of the cycle in $H+T \leq (\log^* n)/3$ rounds, which would be in contradiction with Linial's lower bound~\cite{linial1987LOCAL}. Hence, for each of the cycles, there is some configuration of the IDs that leads to at least one  consecutive set of at least $H\geq \log^* n/5$ nodes being colored all red or all blue. Take one such consecutive set of nodes $H$ that are colored monochromatically, for each cycle. We call these \textit{monochromatic paths}. 
Now, we have $n$ monochromatic paths, one for each cycle, and thus at least $n/2$ of them have the same red or blue color, say blue, without loss of generality. Take $n/H \ll n/2$ of these monochromatic paths, colored blue in their original cycle with certain ID assignments, and append them to each other such that we get a cycle of length $n$. If we run $\mathcal{A}$ on this new cycle, with running time at most $(\log^* n)/100$, for each monochromatic path, only nodes that are within distance at most $(\log^* n)/100$ of the other paths may notice that they are not in their original cycle. Hence, at most $(\log^* n)/50$ nodes switch their cycle per path. That is a total of at most $\frac{n}{\log^* n/5} \cdot \frac{\log^* n}{50} = \frac{n}{10}$ nodes. Hence, we have at most $n/10$ red nodes. Hence, on a certain $n$ node cycle with ID assignments from $\{1, \dots, n^2\}$, algorithm $\mathcal{A}$ fails to compute a coloring with at most $3/4$ of the nodes in each color. Having arrived at a contradiction from the assumption of $\mathcal{A}$ having round complexity $T\leq (\log^* n)/100$, we conclude that any algorithm for balanced coloring (with $3/4$ of the nodes in each color) needs round complexity $\Omega(\log^* n)$. Similar lower bound holds for any other constant balance requirement. 
\end{proof}

\subsection{Incorporating Balanced Coloring in the Algorithm of Rozho\v{n} and Ghaffari}

Next, we show how to incorporate \cref{thm:balanced_coloring} in the algorithm of Rozho\v{n} and Ghaffari~\cite{RozhonG19}. 
This implies the following theorem, which provides a decomposition that, compared to the original algorithm of Rozho\v{n} and Ghaffari, has a much better dependency on the number of bits in the identifiers.

\begin{theorem}\label{thm:ND-coloring-slow}
Consider an arbitrary graph $G$ on $n$ nodes where each node has a unique $b$-bit identifier, where $b=\Omega(\log n)$. 
There is a deterministic distributed algorithm that computes a network decomposition of $G$ with $O(\log n)$ colors and weak-diameter $O(\log^3 n)$ in $O(\log^8 n + (\log^*b) \cdot \log^5 n)$ rounds of the \congest model, using $O(b)$-bit messages. 

Moreover, for each cluster $\fC$ of vertices, we have a Steiner tree $T_\fC$ with radius $O(\log^3 n)$ in $G$, for which the set of terminal nodes is equal to $\fC$. 
Each vertex of $G$ is in $O(\log n)$ Steiner trees of any given color out of the $O(\log n)$ color classes.  
\end{theorem}

\begin{proof}
We show how to incorporate \cref{thm:balanced_coloring} in the algorithm of Rozho\v{n} and Ghaffari~\cite{RozhonG19}. 
Note that their algorithm was explained in \cref{sec:algorithm_intuition}. Recall that in the $i$-th phase of their algorithm, each cluster is given a color based on the $i$-th bit of its identifier. 
After the phase, clusters of different colors are disconnected and will never be connected again. 

Now in each phase $i$, instead of coloring based on the $i$-th bit, we invoke \cref{thm:balanced_coloring} to get a coloring such that in each connected component of clusters consisting of at least two clusters, at most $3/4$ fraction of clusters is colored either blue or red. 
Since at the end of the phase we disconnect all blue clusters from red clusters, each connected component of clusters containing at least $2$ clusters is split into several new connected components, such that the number of clusters in each new connected component is at most $3/4$ of the number of clusters in the original connected component. 
Hence, if we set the number of phases of the algorithm to be $\log_{4/3} n$, at the end of the algorithm, each connected component of clusters contains only one cluster. 

The dependence on the number of bits in the algorithm of Rozho\v{n} and Ghaffari comes from the fact that we need $b$ phases. 
In particular, their algorithm needs $b$ phases, each with $O(b\log n)$ steps, and as such, it computes a $(\log n, b^2 \log n)$ weak-diameter network decomposition in $O(b^4 \log^4 n)$ rounds.

Using the balanced coloring scheme, we can now set $b$ to $\log_{4/3} n$, and thus get a $(\log n, \log^3 n)$ weak-diameter network decomposition in $O(\log^8 n)$ rounds, modulo that we also need to spend $O(\log^3 n\log^* b)$ additional rounds in each phase to compute the coloring, using \cref{thm:balanced_coloring}. 
Hence, the overall round complexity of the algorithm is $O(\log^8 n + (\log^5 n) \cdot \log^* b)$ rounds. 
\end{proof}

\begin{remark}
The above \cref{thm:log5balanced} shows that in order to construct a network decomposition in the \congest model, we do not need to assume $O(\log n)$ bit unique identifiers, but instead it suffices to have a port numbering of the edges and access to an oracle that colors a locally constructed graph of constant degree with constantly many colors. 
\end{remark}

\subsection{Applications in the Shattering Framework}

We now present two corollaries of the above statements that are later improved in \cref{sec:faster_construction_without_identifiers}.

\begin{corollary}
\label{thm:mis_randomized_slow}
There is a randomized distributed algorithm that computes a maximal independent set in $O(\log \Delta \cdot \log\log n + \log^9 \log n)$ rounds of the \congest model, with high probability.
\end{corollary}

\begin{proof}
First, we run the randomized MIS algorithm of Ghaffari~\cite{ghaffari2016MIS} for $O(\log \Delta)$ rounds. As proven in \cite[Lemma 4.2]{ghaffari2016MIS}, this algorithm computes an independent set $S$ such that, after removing all nodes of $S$ and those that have a neighbor in $S$ from the graph, we are left with ``small" connected components, with high probability. Here, small components shows that (A) each component has at most $O(\Delta^4 \log n)$ nodes, (B) any $5$-independent set in each component --- a set where any two nodes have distance at least $5$ --- has size at most $O(\log n)$.

At this point, we run the \congest model randomized ruling set algorithm of Ghaffari~\cite[Lemma 2.2]{ghaffari2019MIS} which computes a $(6, O(\log\log n))$ ruling set of each component,  in $O(\log\log n)$ rounds, with high probability. That is, for each component $\fC$, we get a ruling set $T$ such that (I) each two vertices of the ruling set have distance at least $6$ from each other, (II) each node $v$ in the component knows the closest node $T$ to itself (ties broken arbitrarily) and that node is within distance $O(\log\log n)$. This induces a clustering of the component, i.e., a partitioning of all vertices into disjoint clusters, each with radius $O(\log\log n)$: there is one cluster for each node $v\in T$ and it includes all nodes $u$ in the component for which $v$ is the closest node in $T$ to $u$.

Now, we run the network decomposition algorithm of \Cref{thm:ND-coloring-slow} on the cluster graph where each virtual vertex is a cluster of diameter $O(\log\log n)$ around $u\in T$. 
This runs in $O(\log^9\log n)$ rounds; the additional slowdown of $O(\log\log n)$ comes from the fact that each vertex of the cluster graph is actually a cluster of strong diameter $O(\log\log n)$. 
The fact that the whole construction still works is verified in \cref{rem:change_nodes_to_trees}. 
We get a partition of the cluster graph into vertex-disjoint clusters, each with weak-diameter $O(\log^3\log n)$. In the original graph, this means clusters of weak-diameter $O(\log^4\log n)$, colored with $O(\log\log n)$ colors and such that adjacent clusters have different colors.

We now process the color classes of the network decomposition one by one, and compute the MIS for each of them separately. When we process a color, each cluster of that color works independently, as follows: we first remove nodes of the cluster that already have a neighbor in the MIS. Then, we run $O(\log n)$ independent instances of the MIS algorithm of Ghaffari~\cite{ghaffari2016MIS}, each for $R=O(\log \Delta+\log\log n)$ rounds, on this cluster. We note that since this algorithm works with single-bit messages, we can run $O(\log n)$ independent instances of it in parallel in the \congest model, with no round complexity overhead. The analysis of this algorithm \cite[Theorem 4.2]{ghaffari2016MIS} shows that in each run, each node is either in the computed MIS or has a neighbor in it, with probability at least $1-2^{-\Theta(R)}$. Since the cluster, and even the entire component, has at most $N=O(\Delta^4 \log n)$ nodes, each run succeeds to compute a correct MIS with probability at least $1-N2^{-\Theta(R)} = 1/(\Delta\log n)^{10}$. Then, we locally check each run to see if it produced a correct MIS, again using one-bit messages, so that all runs can be checked in parallel. Finally, we aggregate over a breadth first search tree of the cluster whether each run was successful or not, again using a single bit indicator for each run. Since we have $O(\log n)$ runs, at least one is successful, with high probability. Since the diameter of the cluster is $O(\log^4\log n)$, we can aggregate these indicators in $O(\log^4\log n)$ additional rounds.  We pick one successful run, add the computed MIS to the overall independent set, and we can then proceed to the next color of the decomposition.

Processing each color takes $O(\log \Delta+\log^4\log n)$ rounds. Since we have $O(\log\log n)$ colors in the decomposition, the round complexity of computing the MIS atop the given decomposition is $O(\log \Delta\cdot \log \log n+\log^8\log n)$. We also spent $O(\log^8 \log n)$ rounds to compute the decomposition, which makes the overall round complexity $O(\log \Delta\cdot \log \log n+\log^8\log n)$.
\end{proof}

We get a similar result for $\Delta + 1$ coloring: 
\begin{corollary}
\label{thm:coloring-in-shattering_slow}
There is a randomized distributed algorithm, in the \congest model, that computes a $\Delta+1$ coloring in any $n$-node graph with maximum degree at most $\Delta$ in $O(\log \Delta + \log^9\log n)$ rounds, with high probability.
\end{corollary}

\begin{proof}[Proof sketch]
The proof follows in a similar manner as the proof of \Cref{thm:mis_randomized_slow}, by incorporating the network decomposition that uses balanced coloring into the \congest-model shattering-based coloring algorithm of \cite[Theorem 1.3]{ghaffari2019MIS}.
%\mtodo{It would be good to expand on this and spell out how things go.}
\end{proof}

\section{Faster Identifier-Independent Network Decomposition}
\label{sec:faster_construction_without_identifiers}

In this section, we show how to put the two improvements of \cref{sec:log5} and \cref{sec:handling_identifiers} together. 
The main result is that a network decomposition can be constructed with a round complexity of $O(\log^5 n + \log^4n \log^* b)$.

\subsection{Incorporating balanced coloring in the analysis of \cref{thm:log5}}

Here, we prove a formal version of \cref{thm:log5balanced_informal}. 
\begin{theorem}
\label{thm:log5balanced}
Consider an arbitrary graph $G$ on $n$ nodes where each node has a unique $b$-bit identifier, where $b=\Omega(\log n)$. 
There is a deterministic distributed algorithm that computes a network decomposition of $G$ with $O(\log n)$ colors and weak-diameter $O(\log^2 n)$, in $O(\log^5 n + (\log^*b)\log^4 n )$ rounds of the \congest model with $b$-bit messages. 

Moreover, for each cluster $\fC$ of vertices, we have a Steiner tree $T_\fC$ with radius $O(\log^2 n)$ in $G$, for which the set of terminal nodes is equal to $\fC$. 
Each vertex of $G$ is in at most $O(\log n)$ Steiner trees of each color.
\end{theorem}

\begin{proof}
We explain how to adapt the algorithm from the proof of \cref{thm:log5} by using balanced coloring from \cref{thm:balanced_coloring}. 

We describe what needs to be changed in the description of the algorithm from  \cref{sec:algorithm}. 
First, the number $b$ is not defined as the number of bits, but as $b = 1 + \log_{4/3} n$. 
At the beginning of each phase $i$, clusters in each level $d$ will run the algorithm from \cref{prop:balanced_coloring_insane_version} to compute a partial red and blue coloring of clusters; this has one exception, namely clusters that are in the same level as they were during the previous phase and which were, hence, already considered by the partial coloring of \cref{prop:balanced_coloring_insane_version}. 
These clusters already ran this algorithm for their current level during some previous phase and they retain their color from that previous run (if they were colored).  
The parameter $h_{\text{Prop \ref{prop:balanced_coloring_insane_version}}}$ is set such that $h \log^2 n \ge  200 (\log n + b)^2$. 
%Also, the network $G_{\text{Prop \ref{prop:balanced_coloring_insane_version}}}$ contains only living (i.e., not deleted) vertices. 

The computed color of a cluster $\fC$ plays the same role in this phase as the bit $\ell_{\lev(\fC)+1}$ plays in the original algorithm, i.e., if $u$ and $v$ are neighboring nodes such that clusters $\fC_u$ and $\fC_v$ have the same level, $u$ will consider proposing to $v$ to join $\fC_v$ if $\fC_u$ is colored red and $\fC_v$ is colored blue. 
As we will shortly see, although
\cref{prop:balanced_coloring_insane_version} only outputs a partial coloring, it guarantees that uncolored clusters will not neighbor with a cluster in the same level at any point in time during the current phase, so the fact that not all clusters are colored does not matter for the description of the algorithm. 

We now describe how to adapt the analysis of \cref{thm:log4} to the new algorithm. 
First, the analysis from \cref{sec:analysis}, i.e., the proof of the facts that we delete at most $1/2$ fraction of vertices, the resulting clusters have weak-diameter $O(\log^2n)$ and they have an accompanied Steiner tree of diameter $O(\log^2 n)$ such that each vertex is in $O(\log n)$ Steiner trees, stays the same. 

The round complexity of the network decomposition construction  is $O(\log^5 n + (\log^* b) \cdot \log^4 n)$. 
The first term comes from the analysis of \cref{thm:log5}, while the second term comes from the fact that in each of the $O(\log n)$ phases to construct one of the $O(\log n)$ colors of the resulting decomposition, we need to construct a balanced coloring via \cref{prop:balanced_coloring_insane_version}, with $R = O(\log^2 n)$. 

What remains to be argued is that the resulting clusters are non-adjacent and their Steiner trees are correctly formed, i.e., we will conclude by showing how to adapt the proof of \cref{prop:clusters_nonadjacent} and \cref{prop:trees_are_trees} from \cref{sec:tree}. 
We will slightly change the definition of the transcript tree $T$: each non-leaf vertex of $T$ does not just have $4(b+\log n)$ children, i.e., two times the number of phases, but $6(b+\log n)$, i.e., three times the number of phases. 
After phase $i$, when a cluster $\fC$, mapped to a node $\pi(\fC)$ of $T$ of depth $\lev(\fC)$, decides to go to the next level, we assign it to the $3i$-th, $(3i+1)$-th, or $(3i+2)$-th child of $\pi(\fC)$, based on whether the cluster $\fC$ was assigned a color and if so, which color was assigned to it. 

The proof of \cref{prop:tree_invariant} from \cref{sec:tree} works after the following slight change: We will now observe that if a cluster $\fC$ is left uncolored by \cref{prop:balanced_coloring_insane_version} -- we call such cluster \emph{isolated} --, we know that it will not meet with a different cluster of the same level during the rest of the algorithm (this also shows the algorithm is correctly defined). 
This is because, as we are proving \cref{prop:tree_invariant} by induction, clusters mapped to the subtree of $\pi(\fC)$ in the transcript tree $T$ can, by induction, only eat vertices from clusters in that particular subtree during the next $\left( 2(b+\log n) \right) \cdot \left( 28(b+\log n) \right) \le h \log^2 n / 3 $ phases. 
This means that an isolated cluster $\fC$ can never be adjacent with a cluster on the same level, throughout the whole algorithm. Moreover, once an isolated cluster $\fC$ goes to the next level, it will not eat vertices of other clusters anymore, as it is connected only to clusters of strictly smaller level.
Hence, in future rounds, vertices of isolated clusters can only propose to lower level clusters and join them or be deleted; whenever a lower level cluster neighboring with $\fC$ decides to go to the next level, it deletes its boundary with $\fC$ and does not neighbor with it anymore.
This means that \cref{prop:tree_invariant} holds also in the new algorithm. 
Similarly, the proof of \cref{prop:trees_are_trees} readily generalizes.

Finally, we observe that due to the balanced property of the coloring of \cref{prop:balanced_coloring_insane_version}, whenever new clusters are mapped to some node $r$ in $T$, which happens only once during some phase $i$ of the algorithm, unless all these clusters are isolated, their number is at most $3/4$-th fraction of the clusters in the parent of $r$ at the beginning of the phase $i$. 
Hence, after $1 + \log_{4/3} n$ rounds, all resulting clusters are isolated and, by \cref{prop:clusters_finished}, of level $b$. Hence, there are no edges between the final clusters, as needed. 
\end{proof}

For the shattering applications in the \congest model in \cref{sec:balanced_shattering}, we will need the fact that the above \cref{thm:log5balanced} generalizes to the following, more restrictive, setting. 

\begin{remark}
\label{rem:change_nodes_to_trees}
The above proof of \cref{thm:log5balanced} works even if each node $u$ of the graph $G$ is in the communication graph simulated by a tree of strong-diameter $R_0$. 
The round complexity then changes to $O(R_0 (\log^5 n + (\log^*b)\log^4 n) )$. 
\end{remark}
\begin{proof}
We need to check that both the main network decomposition algorithm from \cref{thm:log5balanced} and the balanced coloring from \cref{prop:balanced_coloring_insane_version} generalize to this more restrictive setting, where each \emph{virtual} node of $G$ is a tree of diameter $R_0$ in the underlying communication graph. 
In the case of \cref{prop:balanced_coloring_insane_version}, we observe that expanding each node $u$ of the Steiner tree $T_{\fC_u}$ of diameter $O(\log^2 n)$, in its underlying tree, makes $T_{\fC_u}$ a tree of diameter $O(R_0 \cdot \log^2 n)$. 
Similarly, we can set the parameter $h$ in \cref{prop:balanced_coloring_insane_version} such that $h\log n = 3R_0\cdot 28(\log n + b)$, i.e., $R_0$ times bigger than its size in \cref{thm:log5balanced}. 
This implies that the resulting coloring will have the desired properties, while its round complexity is still $O((R_0 \cdot R+R_0 \cdot \log^2 n) \cdot \log^* b)$. 

Second, we verify that the network decomposition algorithm generalizes to this setting.
Whenever a cluster $\fC$ collects some information (e.g., the number of proposing vertices) through its Steiner tree, we expand each virtual vertex in it to its corresponding tree and send the information in the new, \emph{expanded Steiner tree} of diameter $O(R_0 \cdot \log^2 n)$. 
Since each virtual node is a part of $O(\log n)$ Steiner trees, each edge in the expanded Steiner trees is a part of $O(\log n)$ expanded Steiner trees. 
This means that gathering information through clusters is done with an additional $R_0$ multiplicative increase in the round complexity. 
Similarly, whenever a virtual node proposes to a cluster, it can decide who to propose to in $O(R_0)$ rounds by gathering information from the leaves of its communication tree. 
Hence, the final round-complexity is multiplied by a factor of $R_0$, which concludes the proof. 
\end{proof}

\subsection{Balanced Coloring for Faster Decomposition}

\begin{proposition}
\label{prop:balanced_coloring_insane_version}
Consider a network $G$ with $b$-bit identifiers, where $b=\Omega(\log n)$, and $O(b)$-bit message sizes. Suppose that the vertices are partitioned into clusters of weak-diameter $R$. 
In particular, for each cluster, we are also given a Steiner tree of depth $R$, such that each vertex is in $O(\log n)$ of these Steiner trees. 
%There are dead nodes that do not participate. 
Furthermore, suppose that each cluster $\fC$ has a level $\lev(\fC) \in [1, O(\log n)]$. 
There is an algorithm that, in $O((R+h\log^2 n) \cdot \log^* b)$ rounds, returns a partial coloring of the clusters with the following guarantees: 

Let $U_i$ be the set of vertices in clusters of level $i$ and define $U_{i+}=\cup_{j\geq i} U_j$. 
We define a cluster graph $G_i$ for each level $i$, where vertices are clusters of level $i$ and two clusters are connected iff their distance in the subgraph of $G$ induced by $U_{i+}$ is at most $h \log^2 n$, for a given value $h \ge 1$. 

In the output partial coloring, each cluster which is contained in a connected component with at least two level-$i$ clusters is colored red or blue such that at most $3/4$ of the level-$i$ clusters are blue and similarly at most $3/4$ of them are red. 
Clusters that are alone in their connected component in $G_i$ are left uncolored. 
\end{proposition}
\begin{proof}
We construct the coloring in parallel for each of the $O(\log n)$ levels. 
First, for each level $i$, we construct in parallel an \emph{extended cluster} of $\fC$ denoted by $\check{\fC}$ as follows. 
We run a simultaneous BFS in $U_{i+}$ starting from all nodes that are contained in some level-$i$ cluster. 
Each level-$i$ BFS only uses a single bit in each $b = \Omega(\log n)$ bit message that can be send across each edge.  
Each node of each level-$i$ cluster starts by sending a one-bit token through the one-bit channel to each of its neighbors in $U_{i+}$. 
%The token carries the level number $i$, as well as the identifier of the cluster, and also the identifier of the initial node in that cluster that started this token. 
In general, we are allowed to forward this level-$i$ token only among nodes of $U_{i+}$. 
Each node $v\in U_{i+}$, upon receiving one (or more) level-$i$ BFS tokens, remembers the first node $w$ it receives a token from as its parent in the BFS tree, breaking ties arbitrarily.  
Moreover, in the next round, $v$ forwards this token to its own neighbors in $U_{i+}$. 
We repeat this for $h\log^2 n$ iterations.  
At the end, each node in $U_{i+}$ that can be reached from a level-$i$ cluster via $h\log^2 n$ hops in $G[U_{i+}]$ is reached and belongs to one level-$i$ BFS. 
Each level-$i$ cluster now has one (potentially singleton) tree $T_u$ attached to each of its vertices $u$, which contains all nodes of $U_{i+}$ that were reached by the token initiated in $u$. 

We define the extended cluster $\check{\fC}$ of $\fC$ as the union of all trees $T_u$ over $u \in \fC$ and the Steiner tree $T_{\check{\fC}}$ as the union of the Steiner tree $T_\fC$ together with trees $T_u$ for $u \in \fC$. 
%Note that all extended clusters can be computed in $O(h \log n)$ rounds for one level-$i$, hence, in $O(h \log^2 n)$ rounds we can compute them for all levels. 
Note that the above construction adds each node to only $O(\log n)$ extended clusters (at most one for each level), hence it is still the case that each vertex is in $O(\log n)$ Steiner trees. 
Thus, using \cref{cor:pipelining_with_overlap}, each cluster can broadcast its label to all its vertices in parallel $O(R + h \log^2 n)$ rounds. 

We can now apply the algorithm from \Cref{thm:balanced_coloring2} for each level-$i$ cluster graph $G_i$ defined such that the nodes are extended clusters of level $i$ and connections are between adjacent clusters. 

Whenever we collect or broadcast an information in a cluster during that algorithm, it can be done for all clusters of all levels in parallel in $O(R + h\log^2 n)$ rounds by \cref{cor:pipelining_with_overlap}, due to the fact that the total number of Steiner trees overlapping at any vertex is $O(\log n)$. 
Whenever we use a particular edge connecting two Steiner trees, it can be used by $O(\log n)$ runs for each level at the same time, hence instead of one \congest round we need $O(\log n)$ rounds. This complexity is, however, dominated by the complexity of broadcasting on Steiner trees, so in total, the round complexity is bounded by $O((R + h \log^2 n) \log^* b) $, as needed. 
\end{proof}

\subsection{Applications in the Shattering Framework}
\label{sec:balanced_shattering}
\cref{rem:change_nodes_to_trees} has the following two corollaries that were mentioned in \cref{sec:our_contributions}. 

\randomizedmis
\begin{proof}
The proof is the same as \Cref{thm:mis_randomized_slow} with only one exception: The $O(\log^9\log n)$ round complexity of building network decomposition is now replaced with an $O(\log^6\log n)$ round complexity, thanks to the faster decomposition provided by \Cref{thm:log5balanced} which needs $O(R_0 \cdot \log^5\log n)$ rounds, where $R_0 = O(\log \log n)$ is the diameter of each cluster formed after construction of the ruling set. 
\end{proof}

Similarly, we get the following improvement for the round-complexity of $\Delta+1$-coloring. 
\randomizedcoloring

\section{Aggregating with Overlapping Trees}
\label{sec:aggregating}

In this section, we explain how to use pipelining to speed up broadcasting and information aggregation in our setting with overlapping broadcast trees. Our end result is \cref{cor:pipelining_with_overlap} that we rely on whenever we want to optimize the round complexity of our algorithms in the \congest model. 

Recall that we face the following problem in several \congest algorithms in this paper. 
We have a collection of rooted trees $T_\fC$ such that the depth of each tree is $R = O(\log^2 n)$ and each edge of the underlying graph $G$ is present in up to $O(\log n)$ trees. 
We now want to solve one of the following two problems: 
\begin{enumerate}
    \item \emph{Broadcast}: The root of $T_\fC$ wants to send an $m = O(\log n)$-bit message to all nodes in $T_\fC$ -- this is useful e.g. when a cluster root tells the vertices in it whether the cluster grows in this step or not;
    \item \emph{Summation}: Each node $u \in T$ starts with a nonnegative $m = O(\log n)$-bit number $x_u$. 
    At the end, the root of $T$ knows the value of $(\sum_{u \in T} x_u) \mod 2^{O(m)}$ -- this is useful e.g. when a cluster collects how many nodes are proposing to it.
\end{enumerate}

For the applications in \cref{sec:handling_identifiers,sec:faster_construction_without_identifiers}, we also need to quickly solve the following two operations:

\begin{enumerate}
    \item \emph{Convergecast}: We have $O(1)$ special nodes $u\in T$, where each special node starts with a separate message. At the end, the root of $T$ knows all messages;
    \item \emph{Minimum}: Each node $u \in T$ starts with a nonnegative number $x_u$. At the end, the root of $T$ should know the value of $\min_{u \in T} x_u$. 
\end{enumerate}

To deal with the overlap of Steiner trees, in case there are $P$ trees using the same edge, we allocate only $b' = b/P$ bits (typically, $b' = \Theta(1)$) of the capacity of each edge to a given tree. 
We then show how to solve the four aforementioned operations on a single tree with $b'$-bit messages in time $O(R + m/b')$, where $m$ is the length of the messages we are transmitting/aggregating. 
Performing broadcast and convergecast operations can be done by ``pipelining'' the messages \cite{peleg00}. 
For example, to perform broadcast of a message of length $m > b'$, the root splits the message into chunks of length $m/b'$ and starts the broadcast of the $i$'th chunk in the $i$'th round. The subsequent broadcasts of different chunks do not interfere, so all of them finish in $O(R + m/b')$ rounds. 
Convergecast is handled similarly. 
To perform summation and taking minimums, each node needs to do a little bit more additional work as explained in the following lemma. 
    
\begin{lemma}
\label[lemma]{lem:congest_aggregation}
Let $T$ be a rooted tree with depth $r$. 
The tree is oriented towards its root and each node knows its parent, as well as its own depth and the overall depth of $T$. 
Moreover, each node $u$ has an $m$-bit number $x_u$. 
In one round of communication, each node can send a $b$-bit message for some $b \le m$ to all its neighbors in $T$. 
There is a protocol such that, in $O(r + m/b)$ rounds, we can perform the following operations
\begin{enumerate}
    \item \emph{Broadcast}: The root of $T$ sends a $m$-bit message to all nodes in $T$;
    \item \emph{Convergecast}: We have $O(1)$ special nodes $u\in T$, where each special node starts with a separate $m$-bit message. At the end, the root of $T$ knows all messages;
    \item \emph{Minimum}: Each node $u \in T$ starts with a nonnegative $m$-bit number $x_u$. At the end, the root of $T$ knows the value of $\min_{u \in T} x_u$;
    \item \emph{Summation}: Each node $u \in T$ starts with a nonnegative $m$-bit number $x_u$. At the end, the root of $T$ knows the value of $(\sum_{u \in T} x_u) \mod 2^{O(m)}$;
\end{enumerate}
\end{lemma}
\begin{proof}
For simplicity, we prove only the case $m=b$, as the generalization to $b \le m$ is direct. 
The Broadcast and Convergecast operations were already sketched above. 

The summation algorithm works as follows: each node $u$ in depth $d$ is sleeping except of rounds $r-d+1$ to $r-d+m$.
The node $u$ starts with a value $x_u$ that will change over time. 
In every round $r-d \le i \le r-d+m$, the node $u$ sends the value of the $(i+d-r)$'th least significant bit $b_u$ to its parent (this only applies if it has a parent and if $i > r-d$) and from each of its children $v$, the node $u$ receives the corresponding value $b_v$ equal to the $(i+(d+1)-r)$'th least significant bit of $x_v$.  
Then, $u$ updates the value of $x_u$ as follows:
\[
x_u \leftarrow x_u - b_u \cdot 2^{i+d-r-1} + \sum_{\text{$v$ child of $u$}} b_v \cdot 2^{i+d-r}.  
\]
Note that after this one-round update, the total sum $\sum_{u \in T} x_u$ did not change. 
On the other hand, we can easily see by induction that after round $i$, each non-root node $u$ in depth $d$ has the $i+d-r$ least significant bits of the value $x_u$ set to zero. 
Hence, after $r + O(m)$ rounds, for each $u$ except the root, we have $\left( x_u \mod 2^{O(m)} \right) = 0$ and, hence, the root has the value $\left( \sum_{u\in T} x_u\right) \mod 2^{O(m)}$, i.e., the final sum.

The case when addition is replaced by taking the minimum (or maximum) is handled similarly, but the nodes start sending the information from the most significant bit. 
More concretely, the algorithm aggregating $\min_{u \in T} x_u$ works as follows: each node $u$ in depth $d$ is sleeping except of rounds $r-d$ to $r-d+m$.
The node $u$ starts with a value $x_u$ and, moreover, it has a bit variable $b_u$ that at the beginning of round $i$ contains the $i+d-r$'th most significant bit of $\min_{w \in T(u)} x_u$ (here, $T(u)$ denotes the subtree of $T$ rooted at $u$). 
The node $u$ also maintains a possibly empty subset $S_u \subseteq \{u\} \cup \bigcup_{\text{$v$ child of $u$}} \{v\}$ such that each child $v$ of $u$ is contained in $S_u$ if and only if, at the beginning of round $i$, the $i+d-r$ most significant bits of $\min_{w \in T(v)} x_w$ are equal to those of $\min_{w \in T(u)} x_w$. Similarly, $u \in S_u$ if and only if the $i+d-r$ most significant bits of $x_u$ are equal to those of $\min_{w \in T(u)} x_w$.  
Initially, we set $S_u = \{u\} \cup \bigcup_{\text{$v$ child of $u$}} \{v\}$ and the variable $c_u$ is first set in round $r-d$. 

In every round $r-d\le i \le r-d+m$, the node $u$ sends the value of the bit $b_u$ to its parent (if it has a parent and if $i > r-d$) and from each of its children $v$, the node $u$ receives the corresponding value $b_v$. 
To update $b_u$ for the next round, the node $u$ considers all values $b_v$ where $v \in S_u$, and the $(i+(d+1)-r)$'th most significant bit of $x_u$ if $u \in S_u$. 
If at least one of those bits is equal to $0$, $b_u$ is set to $0$ and we remove all children $v$ with $b_v = 1$ from $S_u$, as well as $u$ if  the $(i+(d+1)-r)$'th most significant bit of $x_u$ is $1$.  
Otherwise, the value of $b_u$ is set to $1$ and $S_u$ is left the same. 

The correctness of the algorithm follows from the following induction argument. During round $i$, the node $u$ got to know the $i+(d+1)-r$'th most significant bit of all $v \in S_u$ such that $\min_{w \in T(v)} x_w$ and $\min_{w \in T(u)} x_w$ agree on the $i+d-r$ rightmost bits from its children as values $b_v$. 
%Also, it considers the $i+(d+1)-r$'th most significant bit of $x_u$ if it agrees on $i+d-r$ most significant bits with $\min_{w \in T(u)} x_w$. 
The node $u$ then correctly updates $b_u$ as the $i+(d+1)-r$'th most significant bit of $\min_{w \in T(u)} x_w$ and accordingly updates the set $S_u$ afterwards. Hence, after $r+m$ rounds, the root node knows all $m$ bits of the value $\min_{u\in T} x_u$, as needed.  
\end{proof}
\begin{remark}
In general, we are only using the property that the respective operation $\circ$ (such as $+$ or $\min(\cdot,\cdot)$) is associative and if $p_i(x)$ denotes the rightmost (leftmost) $i$ bits of $x$, then $p_i( x_1 \circ \dots \circ x_k)$ can be computed from $p_i(x_1), \dots, p_i(x_k)$. 
For example, multiplication also has this property. 
\end{remark}

The above \cref{lem:congest_aggregation} is used via the following corollary. 

\begin{corollary}
\label{cor:pipelining_with_overlap}
Let $G$ be a communication graph on $n$ vertices. Suppose that each vertex of $G$ is part of some cluster $\fC$ such that each such cluster has a rooted Steiner tree $T_\fC$ of diameter at most  $R$ and each node of $G$ is contained in at most $P$ such trees. 
Then, in $O(P + R)$ rounds of the \congest model with $b$-bit messages for $b \ge P$, we can perform the following operations for all clusters in parallel:
\begin{enumerate}
    \item \emph{Broadcast}: The root of $T_\fC$ sends a $b$-bit message to all nodes in $\fC$;
    \item \emph{Convergecast}: We have $O(1)$ special nodes $u\in \fC$, where each special node starts with a separate $b$-bit message. At the end, the root of $T_\fC$ knows all messages;
    \item \emph{Minimum}: Each node $u \in \fC$ starts with a nonnegative $b$-bit number $x_u$. At the end, the root of $T_\fC$ knows the value of $\min_{u \in \fC} x_u$;
    \item \emph{Summation}: Each node $u \in \fC$ starts with a nonnegative $b$-bit number $x_u$. At the end, the root of $T_\fC$ knows the value of $(\sum_{u \in \fC} x_u) \mod 2^{O(b)}$;
\end{enumerate}
\end{corollary}
\begin{proof} 
Each edge allocates $\lfloor b/P \rfloor$ bits to each Steiner tree that is using it.
Then, for each Steiner tree $T_\fC$ in parallel, we use \cref{lem:congest_aggregation} to perform the given operation in $O(R + b / (b/P)) = O(R + P)$ rounds. 
\end{proof}

\section*{Acknowledgment}
This project has received funding from the European Research Council (ERC) under the European Union’s Horizon 2020 research and innovation programme (grant agreement No. 853109)

\bibliographystyle{alpha}
\bibliography{ref}

\end{document}